 \documentclass[preprint]{alt2026} % camera-ready template (de-anonymized)

\title[Almost Stable Clustering]{The Computational Complexity of Almost Stable Clustering with Penalties}
\usepackage{times}
\altauthor{%
 \Name{Kamyar Khodamoradi} \Email{kamyar.khodamoradi@uregina.ca }\\
  \addr Department of Computer Science, University Of Regina
 \AND
 \Name{Farnam Mansouri} \Email{f5mansou@uwaterloo.ca}\\
 \addr David Cheriton School of Computer Science, University of Waterloo%
 \AND
 \Name{Sandra Zilles} \Email{zilles@cs.uregina.ca }\\
 \addr Department of Computer Science, University Of Regina
}

\usepackage[utf8]{inputenc}     % allow utf-8 input
\usepackage[T1]{fontenc}        % use 8-bit T1 fonts
\usepackage{hyperref}           % hyperlinks
\usepackage{url}                % simple URL typesetting
\usepackage{booktabs}           % professional-quality tables
\usepackage{amsfonts}           % blackboard math symbols
\usepackage{nicefrac}           % compact symbols for 1/2, etc.
\usepackage{microtype}          % microtypography
\usepackage{thmtools,thm-restate}
\usepackage{amsmath,amsfonts,amssymb,mathtools,bbm}

\usepackage{tikz}
\usepackage{algorithm}
\usepackage{algorithmic}
\usepackage{xcolor}
\usepackage{multirow}
\usepackage{enumitem}
\usepackage{xspace}
\usepackage[normalem]{ulem}

\definecolor{darkcyan}{rgb}{0.0, 0.55, 0.55}
\newcommand{\done}{\rlap{$\square$}{\raisebox{2pt}{\large\hspace{1pt}\checkmark}}}

\ifdefined\DEBUG
\newcommand{\attention}[1]{\textcolor{red}{^ #1 ^}} 
\newcommand{\far}[1]{\textcolor{Red}{#1}}
\newcommand{\kam}[1]{\textcolor{WildStrawberry}{#1}}

\newcommand{\farr}[1]{\textcolor{Red}{$\Box$ F: #1}}
\newcommand{\kamr}[1]{\textcolor{WildStrawberry}{$\Box$ K: #1}}
\newcommand{\sanr}[1]{\textcolor{darkcyan}{$\Box$ Sa: #1}}
\newcommand{\farrd}[1]{\textcolor{Red}{$\done$ F: \sout{#1}}}
\newcommand{\kamrd}[1]{\textcolor{WildStrawberry}{$\done$ K: \sout{#1}}}
\newcommand{\sanrd}[1]{\textcolor{darkcyan}{$\done$ Sa: \sout{#1}}}
\else
\newcommand{\attention}[1]{} 
\newcommand{\far}[1]{#1}
\newcommand{\kam}[1]{#1}

\newcommand{\farr}[1]{}
\newcommand{\kamr}[1]{}
\newcommand{\sanr}[1]{}
\newcommand{\farrd}[1]{}
\newcommand{\kamrd}[1]{}
\newcommand{\sanrd}[1]{}
\fi

\newcommand{\cost}{\operatorname{cost}}
\newcommand{\costp}{\cost_{\mathrm{pen}}}

\newcommand{\bR}{\ensuremath{\mathbb{R}}\xspace}
\newcommand{\bS}{\ensuremath{\mathbb{S}}\xspace}

\newcommand{\cF}{\ensuremath{\mathcal{F}}\xspace}

\newcommand{\cI}{\ensuremath{\mathcal{I}}\xspace}

\newcommand{\cO}{\ensuremath{\mathcal{O}}\xspace}

\newcommand{\kms}{\textsc{$k$-Means}\xspace}

\newcommand{\kmd}{\textsc{$k$-Median}\xspace}

\newcommand{\xmd}[1]{\textsc{$#1$-Median}\xspace}
\newcommand{\xms}[1]{\textsc{$#1$-Means}\xspace}
\newcommand{\kcen}{\textsc{$k$-Center}\xspace}
\newcommand{\gti}{\textsc{Grid Tiling Inequality}\xspace}
\newcommand{\pvc}{\textsc{PVC}\xspace}
\newcommand{\pvclong}{\textsc{Partial Vertex Cover}\xspace}

\newcommand{\ETH}{\textsf{ETH}\xspace}
\newcommand{\FPT}{\textsf{FPT}\xspace}
\newcommand{\distbij}[1]{\delta^{bij}_{#1}}

\newcommand{\floor}[1]{\ensuremath{\left\lfloor#1\right\rfloor}}

\newcommand{\bigO}[1]{\ensuremath{\mathcal{O}\left(#1\right)}}
\newcommand{\opt}{\ensuremath{\textsf{OPT}}}

\begin{document}

\maketitle

\begin{abstract}
    We investigate the complexity of stable (or perturbation-resilient) instances of \kms and \kmd clustering problems in metrics with small doubling dimension. While these problems have been extensively studied under multiplicative perturbation resilience in low-dimensional Euclidean spaces (e.g., \citep{friggstad2019exact,C-AS17}), we adopt a more general notion of stability, termed ``almost stable'', \kam{which is closer to the notion of $(\alpha, \varepsilon)$-perturbation resilience introduced by \cite{BL16}}. Additionally, we extend our results to \kms/\kmd with penalties, where each data point is either assigned to a cluster centre or incurs a penalty.

    We show that certain special cases of almost stable \kms/\kmd (with penalties) are solvable in polynomial time. To complement this, we also examine the hardness of almost stable instances and $(1 + \frac{1}{poly(n)})$-stable instances of \kms/\kmd (with penalties), proving super-polynomial lower bounds on the runtime of any exact algorithm under the widely believed Exponential Time Hypothesis (ETH).

\end{abstract}


\section{Introduction}

A fundamental challenge in statistical data analysis is to organize an unlabelled set of data points into groups of similar objects. Clustering, a prominent technique of unsupervised learning, is widely used across scientific disciplines to address this challenge. The computational complexity of various clustering objectives has been extensively studied in the literature. While many well-known objectives, such as \kms, \kmd, and \kcen, are $\mathsf{NP}$-hard to optimize \citep{GK99, Vazirani01, DFKVV04}, simple heuristics like local search have long been known to perform effectively on real-world data. This disparity has sparked significant interest in beyond worst-case analysis, aiming to identify structural properties of the input data points that enable heuristic algorithms to produce nearly optimal solutions. 

\cite{BL12} first introduced the perturbation resilience condition in the context of max-cut clustering, where the maximum cut would remain the unique optimum after multiplying the edge weights by a factor between $1$ and $1 + \gamma$. \cite{ABS10} studied stables instances of centre-based clustering objectives (including \kms and \kmd) and proved that $3$-stable instances of these problems are solvable in polynomial time. Since then, the challenge has been to determine the smallest stability value for which the problem is in $\mathsf{P}$. \cite{BL16} improved the bound to $1 + \sqrt{2}$, and \cite{MMV12} showed that with a slightly weaker promise of metric perturbation resilience, $2$-stable instances of a class of centre-based clustering objectives can be optimized in polynomial time. \cite{friggstad2019exact} focused on metrics of bounded doubling dimension (which include $\bR^d$ for constant $d$) and proved that for any constant $\varepsilon$, $(1 + \varepsilon)$-stable instances of \kms and \kmd are solvable in polynomial time. They complemented this result by showing that when the dimension $d$ is part if the input, a universal constant $\varepsilon_\circ$ exists such that even a PTAS for $(1 + \varepsilon_\circ)$-stable instances of \kms in $\bR^d$ would imply $\mathsf{NP}=\mathsf{RP}$.

The definition of stable clustering problems is motivated by the intuition that most real-world instances of such problems behave in a ``stable'' manner for the right objective \citep{BL12}, i.e., the structure of the optimum solution is resilient against small perturbations in input values. Here, the objective value of a solution is used as a proxy for its resilience to perturbation, i.e., for stable instances the solution with the optimum objective value must maintain its optimality under a set of valid perturbations. To address cases when this model of stability is too restrictive, the notion of $\alpha$-stability can be relaxed  to ensure that perturbing the metric keeps the optimum solution ``$\beta$-close'' (for some notion of closeness) to the optimum solution of the original metric, for some $\beta \geq 0$. \cite{BLP06} viewed stability as a property of the clustering algorithm itself, and allowed for various optimum solutions that are structurally similar under Hamming distance. \cite{BL16} introduced $(\alpha, \beta)$-stable instances of centre-based clustering (see Section \ref{sec:prelim} for details) and provided an approximation scheme for $(\alpha, \beta)$-stable instances of \kmd when $\alpha > 2 + \sqrt{3}$ and $\beta$ is a function of the minimum size of a cluster in the optimal solution.

\subsection{Our Contribution}

We study the perturbation resilience of \kms and \kmd in metrics of bounded doubling dimension. Our ultimate goal is to give tight bounds on the $(\alpha, \beta)$-stability parameter of efficiently solvable problem instances. While we do not fully close the gap, we make significant progress towards this goal by establishing both upper and lower bounds, thereby providing a more refined picture of the complexity of the problem. Our results also extend to the penalty version of \kms and \kmd, in which points may not be assigned to a cluster, thus incurring a penalty. 

%\begin{restatable}[PTAS for Almost Stable \kms]{theorem}{ptasalmoststable}\label{thm:almost_stable}
%    Fix any $d \geq 1$, $\varepsilon' > 0$, and $\beta > 0$. $(1 + \bigO{\varepsilon'})$-approximate solutions to $(1 + \varepsilon', \beta)$-stable instances of the \kms problem can be found in polynomial time when the metric is of doubling dimension $d$ such that for any such solution $S$, we have that $\distbij{\cI}(S, O) \leq \beta$ for any optimal solution $O$.
    
    %Let $\cI=(X, C, \delta)$ be an $(1 + \varepsilon', \beta)$-stable \kms instance, where doubling dimension of $\delta$ is $d$.  Then in polynomial time a $S \in \cF_k$ can be attained such that $\distbij{\cI}(S, O) \leq \beta$ and $\cost(S) \leq (1 + O(\varepsilon')) \cost(O)$ for every $O$ that is an optimum solution of $\cI$.    
%\end{restatable}

We first consider a generalization of $\alpha$-perturbation resilience, where the given instance is not required to have a unique optimum (see Section \ref{sec:prelim} for details). We show that such instances of the \kms and \kmd (with penalties) are solvable over doubling metrics: 

% We first consider a special case of $(\alpha, \beta)$-stability where we set $\beta = 0$. 
% This notion of stability, which we call strongly $\alpha$-stable (see Section \ref{sec:prelim} for details) still generalizes the definition used in recent works, including the one used in \citep{friggstad2019exact}. We show that such instances of the \kms and \kmd (with penalties) are solvable over doubling metrics: 

%\farrd{changed the theorems to not require strong stability}
\begin{restatable}[Exact Solution for 
% Strongly 
Stable \kms]{theorem}{warmup}\label{thm:warmup}
    Fix any $\varepsilon' > 0$.  $(1 + \varepsilon')$-stable instances of \kms in doubling metrics can be solved in polynomial time.
\end{restatable}
\vspace*{-0.5cm}
\begin{restatable}[Exact Solution for 
% Strongly
Stable Penalty \kms]{theorem}{almoststablepen}\label{thm:almost_stable_pen}
   Fix any $\varepsilon' > 0$. $(1 + \varepsilon')$-stable instances of \kms with penalties in doubling metrics can be solved in polynomial time.
\end{restatable}

Next, we focus on the hardness of $(\alpha, \beta)$-instances of the problem (which we call ``almost stable'' instances). Assuming \ETH, we prove that an optimal solution cannot be obtained, even if we allow the algorithm a super-polynomial runtime, and even if we restrict the instance to a bounded-dimensional Euclidean space.

\begin{restatable}[Hardness of Almost Stable Penalty \kmd]{theorem}{kmedpenaltyhardness}\label{thm:k-median-penalty-alpha-beta-stable-hardness-mult-mult}
    For any $\alpha\in(1,1.2)$, $\beta>0$, there is no $f(k) n^{o(\sqrt{k})}$-time algorithm for $(\alpha, \beta)$-stable Euclidean \kmd with penalties 
    % w.r.t. $\CS^{met-pen}_{mult-mult}$ and $\CS^{sol}_{bij}$
    in $\bR^2$, unless \ETH fails.    
\end{restatable}
\vspace*{-0.5cm}
\begin{restatable}[Hardness of Almost Stable \kmd]{theorem}{kmedhardness}\label{thm:k-median-alpha-beta-stable-hardness-mult-mult}
    For any $\alpha\in(1,1.2)$, $\beta>0$,  there is no $f(k) n^{o(\sqrt{k})}$-time algorithm for $(\alpha, \beta)$-stable \kmd 
    % w.r.t. $\CS^{met}_{mult}$ and $\CS^{sol}_{bij}$ in $\bR^3$ 
    in metrics of doubling dimension $3$, unless \ETH fails.
\end{restatable}

Finally, to complement our algorithmic upper bounds for $\alpha$-stable instances of \kms/\kmd (with penalties), we ask the question: \emph{how close to one can $\alpha$ get to ensure the instance is still efficiently solvable?} While we do not obtain \textsf{APX}-hardness, we manage to show that for inverse polynomial values of $\varepsilon$, the problem is not likely to have an \FPT solution.

\begin{restatable}[Hardness of $(1 + \frac{1}{poly(n)})$-Stable Penalty \kmd]{theorem}{kmedpeninversehardness}\label{thm:k-median-alpha-hardness-panlty}
    There is no $f(k) n^{o(k)}$-time algorithm for $(1 + \bigO{\frac{1}{n^{12}}})$-stable Euclidean \kmd with penalties in $\bR^4$, unless \ETH fails. 
\end{restatable}
\vspace*{-0.5cm}
\begin{restatable}[Hardness of $(1 + \frac{1}{poly(n)})$-Stable \kmd]{theorem}{kmedinversehardness}\label{thm:k-median-alpha-hardness}
     There is no $f(k) n^{o(k)}$-time algorithm for $(1 + \bigO{\frac{1}{n^{16}}})$-stable Euclidean \kmd in $\bR^6$, unless \ETH fails.    
\end{restatable}

%%____________________________________________________________________________________________________________________________________________
%%____________________________________________________________________________________________________________________________________________
%%____________________________________________________________________________________________________________________________________________

\section{Preliminaries} \label{sec:prelim}

%\subsection{Stability of \kms / \kmd}

Consider a doubling metric space $(V, \delta)$, where $V$ is a set of points and $\delta$ a distance function over $V$ with doubling dimension $d$. We focus on stable instances of discrete clustering problems $\kms$ and $\kmd$, with and without penalties, over $V$ with metric~$\delta$. Let $X \subseteq V$ be a finite set of data points, $C \subseteq V$ a finite set of candidate centres, $p:  X \to \bR$ a penalty function, $\delta: V \times V \to \bR$ a metric, and $k$ an integer. We denote a \kms/\kmd instance as a triple $(X, C, \delta)$, and a \kms/\kmd with penalties instance as a four-tuple $(X, C, p, \delta)$. A feasible solution to the instance is a set $S \subseteq C$ with $|S|=k$. For \kms (\kmd) the cost of solution is defined as $\cost(S):=\sum_{j \in X} \delta(j, S)^2$ ($\cost(S):=\sum_{j \in X} \delta(j, S)$), and for \kms (\kmd) with penalties the cost of solution is defined as $\costp(S):=\sum_{j \in X} \min\left(\delta(j, S)^2, p(j)\right)$ ($\costp(S):=\sum_{j \in X} \min\left(\delta(j, S), p(j)\right)$), where  $\delta(j, S)$  is defined as $\min_{i \in S} \delta(j, i)$.

Next, we define $\alpha$-stable clustering. Intuitively, $\alpha$-stability for $\kms / \kmd$ ensures that the optimal solution for a perturbed metric remains optimal for the original metric. We follow the concept of multiplicative stability, first introduced by \cite{ABS10} and \cite{BL12}, but we present a more general version here. Let $\alpha \geq 1$. We call an instance $\cI=(X, C, \delta)$ of \kms/\kmd \emph{$\alpha$-stable} if for any $\delta \leq \delta^\prime \leq \alpha \delta$ every optimum solution of the \kms/\kmd (with penalties) instance $\cI^{\prime}=\left(X, C, \delta^{\prime} \right)$ is also an optimum solution for $\cI$. Similarly, an instance $\cI=(X, C, \delta, p)$ of \kms/\kmd with penalties is called  \emph{$\alpha$-stable} if for any $\delta \leq \delta^\prime \leq \alpha \delta$ and any $p'$ with $p \leq p' \leq \alpha^2 \cdot p$ for \kms and $p \leq p^\prime \leq \alpha \cdot p$ for \kmd, every optimum solution of the \kms/\kmd with penalties instance $\cI^\prime=(X, C, \delta^\prime, p^\prime)$ is also an optimum solution for $\cI$. \farrd{removed definition of strongly stable.}
Let $\cost(S)$ and $\cost'(S)$ denote the cost of $S$ before and after perturbation. We also study $(\alpha, \beta)$-stable instances of \kms and \kmd, as defined by \cite{BL16}. Recall that in this generalization, $\beta$ is a measure of proximity between feasible solutions. Various distance metrics can be considered to define $\beta$-closeness. We adopt the notion $\distbij{\cI}: \binom{C}{k} \times \binom{C}{k} \to \bR^{\geq 0}$, defined as follows. However, all our results  easily extend to other natural notions of distance. For any $S_1, S_2 \subseteq C$ with $|S_1| = |S_2| = k$, define: 
\begin{equation}
    \distbij{\cI}(S_1, S_2) := \min \left\{ \sum\nolimits_{s_1 \in 
S_1 } \delta(s_1, f(s_1)) \mid f \text{ is a bijection between $S_1$ and $S_2$} \right\}.
\end{equation}
We also generalize this notion to clustering with penalties. Let $\alpha \geq 1$ and $\beta \geq 0$. We call an instance $\cI=(X, C, \delta)$  ($\cI=(X, C, \delta, p))$ of \kms/\kmd (with penalties) $(\alpha, \beta)$-stable if for any optimum solution $O \subseteq C$ of $\cI$, and any optimum solution $O^\prime \subseteq C$ of the \kms/\kmd instance (with penalties) $\cI^{\prime}=\left(X, C, \delta^{\prime}\right)$ ($\cI^\prime=(X, C, \delta^\prime, p^\prime))$, where $\delta^\prime$ is any symmetric function such that $\delta \leq \delta^\prime \leq \alpha \delta$ (and $p'$ is any function such that $p \leq p' \leq \alpha^2 \cdot p$ for \kms, and $p \leq p^\prime \leq \alpha \cdot p$ for \kmd), we have $\distbij{\cI}(O, O^\prime) \leq \beta$.

% \begin{definition}[$(\alpha, \beta)$-stability] 
% % Let $\CS^{met}: \cF \times \bR \to 2^\cF$, $\CS^{met-pen}: \cF \times \cG \time \bR \to 2^{\cF} \times 2^{\cG}$, $\CS^{sol}:2^{2^{\binom{C}{k}}} \times \bR \to 2^{2^{\binom{C}{k}}}$. 
% Let $\alpha \geq 1$ and $\beta \geq 0$.
% % Consider any instance of \kms (\kmd) $\cI=(X, C, \delta)$ with its set of optimal solutions being $\overline{O}$. 
% We call an instance $\cI=(X, C, \delta)$  ($\cI=(X, C, \delta, p))$ of \kms/\kmd with penalties) $(\alpha, \beta)$-stable if for any optimum solution $O \subseteq C$ of $\cI$, and any optimum solution $O^\prime \subseteq C$ of the \kms/\kmd instance (with penalties) $\cI^{\prime}=\left(X, C, \delta^{\prime}\right)$ ($\cI^\prime=(X, C, \delta^\prime, p^\prime))$, where $\delta^\prime$ is any symmetric function such that $\delta \leq \delta^\prime \leq \alpha \delta$ (and $p'$ is any function such that $p \leq p' \leq \alpha^2 \cdot p$ \farr{$p \leq p^\prime \leq \alpha \cdot p$ for \kmd}), we have $\distbij{\cI}(O, O^\prime) \leq \beta$.
% \end{definition}

Observe that setting $\beta = 0$ yields the definition of $\alpha$-stable instances. If $\beta > 0$, we call the instance of \kms/\kmd  an \emph{almost stable} instance. Below, we study the performance of a local search algorithm for cases where $\alpha$ is close to 1, that is $\alpha = 1 + \varepsilon'$. To avoid confusion, we  use $\varepsilon'$ to refer to the stability parameter, and $\varepsilon$ to refer to a parameter of our local search analysis. It should be mentioned that $\varepsilon$ only affects the runtime of the local search and is a function of $\varepsilon'$.

\section{Polynomial Algorithms for Stable Clusterings} \label{sec:pol}
In this section, we build upon the exact polynomial-time algorithm of \cite{friggstad2019exact} \farrd{made some minor changes here}
% in two ways: 
% \emph{i)} extending it to strongly stable instances, and \emph{ii)} 
and design a version that incorporating penalties. Despite high-level similarities, our analysis departs where needed to handle penalty terms. We first show how to obtain optima for $\kms$ without penalties in doubling metrics (see Theorem~\ref{thm:warmup}). Next, we provide a polynomial-time algorithm for $\kms$ with penalties (see Theorem~\ref{thm:almost_stable_pen}). Let $(X, C, p, \delta)$ be the given instance of  \kms with penalties. Define $\cF_k = \{S \subseteq C:|S|=k\}$ to be the set of feasible solutions. Throughout this section, we use Algorithm~\ref{alg:rho-swap}, a generalization of the local search method introduced by \cite{friggstad2019exact} to $\kms$ with penalties. Their original algorithm is simply the standard $\rho$-swap local search algorithm, modified to perform the swap that provides the best improvement at each step. We set $\rho=d^{O(d)} \cdot \varepsilon^{-O(d / \varepsilon)}$.

\begin{algorithm} 
\caption{$\rho$-Swap Local Search}
\textbf{Input:} Let $\cF_k$ be the family of feasible sets. Let $S$ be any set in $\cF_k$. \\
\textbf{Output:} The local optimal set $S$

\begin{algorithmic}[1] \label{alg:rho-swap}
\WHILE {there exists a set $S' \in \cF_k$ such that $|S - S'| \leq \rho$ and $\cost(S') < \cost(S)$ ($\costp(S') < \costp(S)$)}
    \STATE $S \gets \arg \min_{S' \in \cF_k, |S - S'| \leq \rho} \cost(S')$ ($\arg \min_{S' \in \cF_k, |S - S'| \leq \rho} \costp(S')$ )
\ENDWHILE
\STATE \textbf{return} $S$
\end{algorithmic}
\end{algorithm}

% We start by introducing some terminology originally established by \cite{friggstad2019exact}, which we have expanded to facilitate us for the cases were penalties exist. 
To expand Friggstad et al.'s terminology to the case with penalties, %We expand the terminology of \cite{friggstad2019exact} to the case with penalties. We also prove slightly altered versions of several of their lemmas and theorems. % proven in \cite{friggstad2019exact}.
% Let $O \in \cF_k$ be the unique optimum solution. 
for any $S, O \in \cF_k$, define:
\begin{itemize}
    \item For $j \in X$,
    % if there exist a centre in $S$ with distance less than $p(j)$ 
    let $\sigma(j, S)$ be the centre in $S$ nearest to  $j$, breaking ties by a fixed ordering of $C$. 
    % Otherwise, let $\sigma(j, S)$
    \item $\overline{X}_{S, O}=\{j \in X: \sigma(j, S) \in S-O$ and $\sigma(j, O) \in O-S\}$. 
    \item $\overline{X}^{pen}_{S, O}=\{j \in X: \sigma(j, S) \in S-O$ and $\sigma(j, O) \in O-S$ and $\delta(j, S)^2, \delta(j, O)^2 < p(j)$$\}$.
    \item $\Psi(S, O)=\sum_{j \in \overline{X}_{S, O}} \delta(j, \sigma(j, S))^2+\delta(j, \sigma(j, O))^2$, $\Psi_{pen}(S, O)=\sum_{j \in \overline{X}^{pen}_{S, O}} \delta(j, \sigma(j, S))^2+\delta(j, \sigma(j, O))^2$.
    % \item $S' = S \backslash O$, $O' = O \backslash S$.
\end{itemize}

\cite{friggstad2019exact} first showed that Algorithm~\ref{alg:rho-swap} visits a ``nearly-good'' solution in a polynomial number of rounds and returns such a solution upon termination. Next, they used the properties of stable instances to show that any nearly-good solution must be the unique optimal solution, hence Algorithm~\ref{alg:rho-swap} is an exact polynomial-time algorithm to solve the problem. We slightly modify the notion of nearly-good solutions to handle technical aspects of our generalized definition of stability. 
% We also introduce a weaker version and a natural extension to the setting with penalties: 

%However, to accommodate the more general concept of $(\alpha, \beta)$-stability, we impose additional restrictions on the set of nearly-good solutions. Unlike \cite{friggstad2019exact}, who require \san{the condition} \sanr{(phrasing is a bit unclear, since here the reader doesn't know what ``the condition'' refers to.)} to be satisfied only for the unique solution of the \kms instance, we require it for every $O \in \cF_k$.

\farrd{We only need one nearly-good definition. I changed the definition.}

\begin{definition}\label{def:nearly-good}
    We call $S \in \cF_k$ a \emph{nearly-good} solution if for every optimum solution $F \in \cF_k$ we have $\cost(S) \leq \cost(F)+2 \varepsilon \cdot \Psi(S, F)$ for $\kms$ and $\costp(S) \leq \costp(O)+ 2 \varepsilon \cdot \Psi_{pen}(S, F)$ for $\kms$ with penalties. 
\end{definition} 

% \begin{definition}\label{def:nearly-good}
%     We call $S \in \cF_k$ a \emph{nearly-good} solution if for every solution $F \in \cF_k$ we have $\cost(S) \leq \cost(F)+2 \varepsilon \cdot \Psi(S, F)$ for $\kms$ and $\costp(S) \leq \costp(O)+ 2 \varepsilon \cdot \Psi_{pen}(S, F)$ for $\kms$ with penalties. In addition, let $\cO$ denote the set of optimal solutions for the instance. We call $S \in \cF_k$ \emph{weakly nearly-good} if there is a solution $O \in \cO$ such that $\cost(S) \leq \cost(O)+ 2 \varepsilon \cdot \Psi(S, O)$ for \kms and $\costp(S) \leq \costp(O)+ 2 \varepsilon \cdot \Psi_{pen}(S, O)$ for \kms with penalties.
% \end{definition} 

% Note that \cite{friggstad2019exact} \kam{use a similar} definition but only requiring the condition to be satisfied for the unique solution of \kms. We require it to be satisfied for every $O \in \cF_k$. In other words, our condition for nearly-goodness is stricter than  that of \cite{friggstad2019exact}. 

Using a by now standard technique introduced by \cite{friggstad2019local}, we can show the existence of a desirable partitioning of the set of centres via a probabilistic argument. Fix any $S', O', T \subseteq C$ and $X' \subseteq X$ such that $S'$ and $O'$ are disjoint, $T \subseteq O' \cup S'$, and $|T| = |O'| = |S'|$. Define $\Delta_j^T$ for each $j \in X'$ to be $\delta\left(j, \sigma\left(j, S^{\prime} \triangle T\right)\right)^2-\delta\left(j, \sigma\left(j, S^{\prime}\right)\right)^2$. $\Delta_j^T$ denotes the change in $j$'s assignment cost when solution $S'$ is replaced by $S' \Delta T$. One obtains the following.

\begin{theorem}[Theorem 5 in \citep{friggstad2019exact}] \label{thm:rand-partition}
    There is a randomized algorithm that samples a partition $\pi$ of $S^{\prime} \cup O^{\prime}$ such that $\left|T \cap S^{\prime}\right|=\left|T \cap O^{\prime}\right| \leq \rho$ for each $T \in \pi$ and $$\mathbf{E}_\pi\left[\sum\nolimits_{T \in \pi} \sum\nolimits_{j \in X'} \Delta_j^T\right] \leq \sum\nolimits_{j \in X'}(1+\varepsilon) \cdot \delta\left(j, O^{\prime}\right)^2-(1-\varepsilon) \cdot \delta\left(j, S^{\prime}\right)^2\,.$$
\end{theorem}

%Note that Theorem~\ref{thm:cost-drop} was originally stated for a unique optimum $O$ by \cite{friggstad2019exact}. However, nowhere in the proof was the uniqueness or optimality used. In fact, Theorem~\ref{thm:cost-drop} is a result of using Theorem~\ref{thm:rand-partition} with $S' = S \backslash O$, $O' = O \backslash S$ and $X' = \overline{X}_{S}$ (see Theorem 6 of \cite{friggstad2019exact}).

%\kam{The following theorem is proved in \cite{friggstad2019exact} for the unique optimum solution $O$. However, the proof does not rely on the uniqueness or optimality of $O$. Here, we restate it for any feasible solution $S$ and $O$}:

Next, we observe that if a feasible solution $S$ is not nearly-good according to our more strict definition, it still holds that the local search algorithm must take a large step on $S$:

\begin{theorem}[Theorem 6 in \citep{friggstad2019exact})] \label{thm:cost-drop}
    For any $S, O \in \cF_k$, if $\cost(S)>\cost(O)+\varepsilon \cdot \Psi(S, O)$, then there exists an $S^{\prime} \in \cF_k$ with $\left|S-S^{\prime}\right| \leq \rho$ and $\cost\left(S^{\prime}\right) \leq \cost(S)+\big(\cost(O)-\cost(S)$ $+ \varepsilon \cdot \Psi(S, O)\big)/k$.
% $$
% \cost\left(S^{\prime}\right) \leq \cost(S)+\frac{\cost(O)-\cost(S)+\varepsilon \cdot \Psi(S, O)}{k}
% $$
\end{theorem} 

Rather than recreating Friggstad et al.'s proof of the theorem, we simply observe that at no point does their proof require $O$ to be globally optimal or $S$ to be locally optimal.  Theorem~\ref{thm:cost-drop} implies that Algorithm~\ref{alg:rho-swap} converges to a 
% weakly 
nearly-good solution in  polynomially many iterations. Due to space constraints, we defer the proof to the appendix.

\begin{lemma} \label{lem:rho-swap-terminate}
     Given any instance of the \kms problem, Algorithm~\ref{alg:rho-swap} with $\rho=d^{O(d)} \cdot \varepsilon^{-O(d / \varepsilon)}$ terminates at a nearly-good solution. Also, within $2 k \cdot \ln (n \Delta)$ iterations, the algorithm produces a nearly-good solution, where $\Delta=\max _{j \in X, i \in C} \delta^2(i, j)$.
\end{lemma}

\subsection{Stable \kms without Penalties} \label{sec:pol-no-penalty}

Lemma~\ref{lem:rho-swap-terminate} alone does not guarantee that Algorithm~\ref{alg:rho-swap} finds an optimal solution for the $(1 + \varepsilon')$-stable instances of \kms. We use the stability condition to show that a
% weakly
nearly-good solution must be one of the optimal ones for the instance. Combined with the fact that the returned solution of Algorithm~\ref{alg:rho-swap} is nearly-good 
% (which is also weakly nearly-good),
we get the desired result:

\warmup*

\begin{proof} 
    \farr{I changed the proof here to not require strong stability.}
    \kam{We refrain from recreating the proof in entirety. Instead we refer interested readers to the proof of Lemma 2 in \citep{friggstad2019exact}.}\kamr{I would say we should omit this since we are invoking the needed lemmas from \cite{friggstad2019exact} anyway.}. We first perturb the distances as follows:

    \begin{equation}\label{eq:perturbed-dist}
        \delta^{\prime}(i, j)=\left\{\begin{tabular}{l l}
            $\left(1+\varepsilon^{\prime}\right) \cdot \delta(i, j)$    &   ; if $i \neq \sigma(j, S)$ \\
            $\delta(i, j)$                                              &   ; otherwise
    \end{tabular}\right.
    \end{equation}    

    We present a slightly modified version of Lemma 2 in \citep{friggstad2019exact}:
    
    \begin{lemma}[Lemma 2 of \citep{friggstad2019exact}, Restatement]
        Let $S, O \in \cF_k$ be feasible solutions. If $cost(S) \leq cost(O) + 2\varepsilon \Psi(S, O)$, then $cost'(S) \leq cost'(O)$.\sanrd{The sentence still isn't grammatical. There are still two ``if''s. How about: ``Let $S, O \in \cF_k$ be feasible solutions. If $cost(S) \leq cost(O) + 2\varepsilon \Psi(S, O)$, then $cost'(S) \leq cost'(O)$.''?}
    \end{lemma}    
    
    In this lemma, we let $S$ be the nearly-good solution of Algorithm~\ref{alg:rho-swap}. Let $O$ be any optimum solution of $\cI^\prime$. Using stability condition $O$ should also be an optimum solution for $\cI$. We observe that
    \[ \cost'(S) = \cost(S) \leq \cost'(O),\]
    where $\cost'$ indicates the cost of a solution under $\delta'$. Therefore $S$ is a optimum solution for $\cI^\prime$. Using the stability condition one more time, we conclude that $S$ must also be an optimum solution of $\cI$ and this $S \in \cO = \{O_1, \ldots, O_\ell\}$. Assume that Algorithm~\ref{alg:rho-swap} encounters $S$ in iteration $t \leq \floor{2k\cdot \ln (\Delta n)}$. Since $S$ has the lowest possible cost for the instance $\cI$, the local search must terminate at iteration $t$ and return $S$. 
\end{proof}

\subsection{Stable \kms with Penalties}\label{sec:pol-penalty}

In this section, we provide polynomial-time exact ($(1 + \bigO{\varepsilon'})$-approximate) solutions for 
% strongly 
stable instances of \kms with penalties in doubling metrics. All the results in this section are extendable to \kmd with minor modification. We first show how to solve $(1 + \varepsilon')$-stable instances of \kms exactly in polynomial time when the metric has a bounded doubling dimension. We first focus on $(1 + \varepsilon')$-stable instances for which any optimal solution will remain the optimum for the perturbed instance.

%First, we focus on extending the result of \cite{friggstad2019exact} to \kms with penalties by attaining the exact solution of $(1 + \varepsilon^\prime)$-stable \kms with penalties instances in polynomial time. 
% Note that we are first focusing on $\CS^{met-pen}_{mult-mult}$ rather than $\CS^{met-pen}_{mult-fix}$ as it is a stricter condition, and so it's an easier problem.

% \begin{theorem} \label{thm:k-means-alpha-solution-penalty-polynomial}
%    Let $\cI=(X, C, P, \delta)$ be an $1 + \varepsilon'$ stable \kms where the dimension of doubling $\delta$ is $d$. Then $\cI$ can be solved in polynomial time.
% \end{theorem}

\almoststablepen*

The local search in \citep{friggstad2019exact} is oblivious to penalties and cannot guarantee nearly-good solutions under this new setting.  Our core idea for addressing this issue is to capture the penalties for the points of $X$ by adding a dummy centre $z^*$ to $V$ (specifically, to $C$) and defining $C'=C \cup \{z^*\}$. We extend $\delta$ to $V \cup \{z^*\}$ by setting $\delta(z^*, j)^2 = \delta(j, z^*)^2 = p(j)$ for every $j \in X$, and $\delta(z^*, c) = 0$ for all $c \in C'$. Note that $\delta$ no longer satisfies the triangle inequality, making the analysis more challenging. For any solution $S$ of the \xms{k+1} instance $(X, C^{\prime}, \delta)$ that contains $z^*$, the cost of $S$ equals the cost of $S \backslash \{z^*\}$ for the original \kms with penalties instance $(X, C, p, \delta)$. Thus, the task of finding the optimum solution of the $(X, C, p, \delta)$ instance is reduced to finding the optimum solution that contains $z^*$ for the \xms{k+1} instance $(X, C^{\prime}, \delta)$. Our key insight here is to modify the $\rho$-swap algorithm to ensure it always makes the cheapest swap that does not evict $z^*$. This eliminates the need to invoke the triangle inequality on $z^*$.

We first show that our modified local search converges to a nearly-good solution in polynomial time. Then, we prove that the stability condition ensures that any such nearly-good solution is, in fact, the optimal one. While the structure of our proof is similar to Friggstad et al.'s, we need to tackle challenges in each step. We will highlight how we address these challenges in what follows.

% Now $k+1$-MEANS problem with this new set of candidate centres would be equivalent with \kms problem with the original set of centres with penalties. However, with the addition of the new centre the triangle inequality breaks down. However, since every $S$ also contains $z^*$, i.e., $z^*$ does not get swapped, we show that this does not matter. Next, we prove this idea rigorously.

Before rigorously proving Theorem~\ref{thm:almost_stable_pen}, %\ref{thm:k-means-alpha-solution-penalty-polynomial}, 
we prove that a nearly-good solution can be found in polynomial time for $\kms$ with penalties as well. To do so, it suffices to replicate Theorem~\ref{thm:cost-drop} for the penalty setting. Consider the original instance $(X, C, p, \delta)$ and fix any $S, O \in \cF_k$. Define $S^{\prime} = S \backslash O$ and $O^{\prime} = O \backslash S$, and fix some $T \subseteq O^{\prime} \cup S'$ such that $|T| = |O'| = |S'|$. Moreover, for any $j \in \overline{X}^{pen}_{S, O}$ define $\Tilde{\Delta}_j^T$ to be $\min(\delta\left(j, \sigma\left(j, S^{\prime} \triangle T\right)\right)^2, p(j))-\delta\left(j, \sigma\left(S^{\prime}\right)\right)^2$. Using Theorem~\ref{thm:rand-partition} with $X' = \overline{X}^{pen}_{S, O}$, there exists a randomized algorithm that samples a partition $\pi$ of $S^{\prime} \cup O^{\prime}$ with $\left|T \cap S^{\prime}\right|=\left|T \cap O^{\prime}\right| \leq \rho$ for each $T \in \pi$ such that
\begin{equation} \label{eq:rand-partition-penalty}
   \begin{aligned}
    \mathbf{E}_\pi\left[\sum\nolimits_{T \in \pi} \sum\nolimits_{j \in \overline{X}^{pen}_{S, O}} \tilde{\Delta}_j^T\right] &\leq 
     \mathbf{E}_\pi\left[\sum\nolimits_{T \in \pi} \sum\nolimits_{j \in \overline{X}^{pen}_{S, O}} \delta\left(j, \sigma\left(j, S^{\prime} \triangle T\right)\right)^2 -\delta\left(j, \sigma\left(j, S^{\prime}\right)\right)^2\right] \\
     &\leq \sum\nolimits_{j \in \overline{X}^{pen}_{S, O}}(1+\varepsilon) \cdot \delta\left(j, O^{\prime}\right)^2-(1-\varepsilon) \cdot \delta\left(j, S^{\prime}\right)^2
\end{aligned} 
\end{equation}

We defer the proof of the following theorem to the appendix.

\begin{theorem} \label{thm:cost-drop-penlty}
    For any $S, O \in \cF_k$, if $\costp(S)>\costp(O)+\varepsilon \cdot \Psi_{pen}(S, O)$, then there exists an $S' \in \cF_k$ with $\left|S-S^{\prime}\right| \leq \rho$ such that
$$
\costp\left(S^{\prime}\right) \leq \costp(S)+\frac{\costp(O)-\costp(S)+\varepsilon \cdot \Psi_{pen}(S, O)}{k}
$$
\end{theorem} 

Now similarly to how one obtains  Lemma~\ref{lem:rho-swap-terminate}, we can derive the following lemma.

\begin{lemma} \label{lem:rho-swap-terminate-penalty}
    For \kms with penalties when Algorithm~\ref{alg:rho-swap} terminates, the returned solution is a nearly-good solution. Also, within $2 k \cdot \ln (n \Delta)$ iterations Algorithm~\ref{alg:rho-swap} will have had some iteration with $S$ being a 
    % weakly 
    nearly-good solution, where $\Delta=\max _{j \in X, i \in C} \delta^2(i, j)$.
\end{lemma}

\farrd{I have edited the following proof to not require strong stability.}
\paragraph{Proof of Theorem~\ref{thm:almost_stable_pen}}%\ref{thm:k-means-alpha-solution-penalty-polynomial}}
Let $\varepsilon$ be such that $1+6 \varepsilon=\left(1+\varepsilon^{\prime}\right)^2$, roughly speaking we have $\varepsilon \approx \varepsilon^{\prime} / 3$ for small $\varepsilon^{\prime}$. Moreover, let $S$ be the 
% weakly-
nearly-good solution that Algorithm~\ref{alg:rho-swap} encounters within $2 k \cdot \ln (n \Delta)$ iterations (by Lemma~\ref{lem:rho-swap-terminate-penalty}).
% , and let $O_q \in \cO$ be the optimal solution for which $S$ is nearly-good. 
Define a perturbed distance function $\delta^{\prime}$ as in \eqref{eq:perturbed-dist}
% \begin{equation}\label{eq:perturbed-dist}
%         \delta^{\prime}(i, j)=\left\{\begin{aligned}
%     \left(1+\varepsilon^{\prime}\right) \cdot \delta(i, j) & \text { if } i \neq \sigma(j, S)  \\
%     \delta(i, j) & \text { otherwise }
%     \end{aligned}\right.
%     \end{equation}
and a perturbed penalty function via
\begin{equation} \label{eq:perturbed-pen}
    p^{\prime}(j) = \left\{\begin{tabular}{l l}
        $\left(1+\varepsilon^{\prime}\right)^2 p(j)$\   & ; if $\delta(j, S) < p(j)$\,,  \\
        $p(j)$                                          & ; otherwise 
    \end{tabular}\right.
\end{equation}
Moreover, for any $S^{\prime} \in \cF_k$, let $\costp^{\prime}\left(S^{\prime}\right)=\sum_{j \in X} \min(\min _{i \in S^{\prime}} \delta^{\prime}(i, j)^2, p^{\prime}(j))$ be the cost of $S^{\prime}$ under  $\delta^{\prime}$. Clearly, 
% $(\delta^{\prime}, p') \in \CS^{met-pen}_{mult-mult}(\delta, P, 1 + \varepsilon^{\prime})$. 
$\delta \leq \delta^\prime \leq (1 + \varepsilon^\prime) \cdot \delta$ and $p \leq p^\prime \leq (1 + \varepsilon^\prime)^2 \cdot p$. Let $O_q$ be any optimum solution w.r.t. $\costp^{\prime}$, due to stability condition $O$ should be an optimum solution w.r.t. $\costp$. Therefore, $S$ is nearly good for $O_q$.
Partition $X$ as follows: 
\begin{itemize}
    \item $X^1=\{j \in X: \sigma(j, S) \in S-O_q$ and $\sigma(j, O_q) \in S \cap O_q$ and $\delta(j, S)^2 < p(j)\}$
    \item $X^2=\{j \in X: \sigma(j, S) \in S \cap O_q$ and $\sigma(j, O_q) \in O_q-S$ and $\delta(j, O_q)^2 < p(j)\}$
    \item $X^3=\{j \in X: \sigma(j, S), \sigma(j, O_q) \in S \cap O_q$ or $\delta(j, S)^2, \delta(j, O_q)^2 \geq p(j)\}$
    \item $X^4=\{j \in X: \sigma(j, S) \in S-O_q$ and $\sigma(j, O_q) \in O_q-S$ and $\delta(j, S)^2, \delta(j, O_q)^2 < p(j)\}$.
\end{itemize}
Observe that$X^4=\overline{X}^{pen}_{S, O_q}$. As in the proof of Theorem~\ref{thm:cost-drop-penlty}, let $c_j^*=\min\{\delta(j, \sigma(j, O_q))^2, p(j)\}$ be the cost incurred by point $j$ in the optimum solution, and analogously  $c_j=\min\{\delta(j, \sigma(j, S))^2, p(j)\}$ for $S$. To calculate $\costp^{\prime}(O_q)$, we consider each $j \in X$ case by case.

\textbf{i)} For $j \in X^1$ or $j \in X^4$ we have $p'(j) = (1 + \varepsilon^{\prime})^2 p(j)$ and $\delta^{\prime}(i, j) = (1 + \varepsilon^{\prime}) \delta(i, j)$ for all $i \in O_q$. Thus, $\min\{\min _{i \in O_q} \{\delta^{\prime}(i, j)^2\}, p^{\prime}(j)\} = (1 + \varepsilon^{\prime})^2 c^*_j$. \textbf{ii)} For $j \in X^2$, only $i = \sigma(j, S) \in O_q$ does not have its $\delta^{\prime}(i, j)^2$ multiplied by $(1 + \varepsilon^{\prime})$.  Thus, if $\delta^2(j, S) \leq p(j)$ we have $\min\{\min _{i \in O_q} \{\delta^{\prime}(i, j)^2\}, p^{\prime}(j)\} = \min\{ (1 + \varepsilon^{\prime})^2 c^*_j, \delta(j, S)\}$. Otherwise, we have that $\min\{\min _{i \in O_q} \{\delta^{\prime}(i, j)^2\}, p^{\prime}(j)\} = \min\{ (1 + \varepsilon^{\prime})^2 c^*_j, p(j)\}$. Therefore, $\min\{\min _{i \in O_q} \delta^{\prime}(i, j)^2, p^{\prime}(j)\} = \min\{ (1 + \varepsilon^{\prime})^2 c^*_j, c_j\}$. \textbf{iii)} For $j \in X^3$ clearly $\min\{\min _{i \in O_q} \delta^{\prime}(i, j)^2, p^{\prime}(j)\} = c_j = c^*_j$. Thus we derive

\begin{equation} \label{eq:k-means-alpha-solution-penalty-polynomial-1}
    \begin{aligned}
\costp^{\prime}(O_q)=\sum\nolimits_{j \in X^1} \left(1+\varepsilon^{\prime}\right)^2 \cdot c_j^* +\sum\nolimits_{j \in X^2} \min \left\{\left(1+\varepsilon^{\prime}\right)^2 \cdot c_j^*, c_j \right\} \\
+\sum\nolimits_{j \in X^3} c_j^*+\sum\nolimits_{j \in X^4} \left(1+\varepsilon^{\prime}\right)^2 \cdot c_j^*.
\end{aligned}
\end{equation}

Finally, we make one last observation, the proof of which is deferred to the appendix.
\begin{lemma}\label{lem:c4}
If $S$ is nearly-good for $O_q$, then $\sum_{j \in X^4} c_j 
%\leq \frac{1}{1-2 \varepsilon} \left(\sum_{j \in X^1} c_j^*+(1+2 \varepsilon) \cdot \sum_{j \in X^4} c_j^*\right) 
\leq(1+6 \varepsilon) \cdot\left(\sum_{j \in X^1} c_j^*+\sum_{j \in X^4} c_j^*\right)$.
\end{lemma}
% similarly to \cite{friggstad2019exact}. As $S$ is a nearly-good solution, $c_j^* \leq c_j$ for $j \in X^2$, and $c_j^*=c_j$ for $j \in X^3$, we have:
% \begin{equation*}
% \begin{aligned}
% \sum\nolimits_{j \in X^4} c_j & \leq \sum\nolimits_{j \in X^1} c_j+\sum\nolimits_{j \in X^4} c_j=\costp(S)-\sum\nolimits_{j \in X^2} c_j-\sum\nolimits_{j \in X^3} c_j \\
% & \leq \costp(S)-\sum\nolimits_{j \in X^2} c_j^*-\sum_{j \in X^3} c_j^* \\
% & \leq \costp(O_q)+2 \varepsilon \cdot \Psi_{pen}(S, O_q)-\sum\nolimits_{j \in X^2} c_j^*-\sum\nolimits_{j \in X^3} c_j^* \\
% & =\sum\nolimits_{j \in X^1} c_j^*+\sum_{j \in X^4} c_j^*+2 \varepsilon\left(\sum\nolimits_{j \in X^4} c_j^*+c_j\right)
% \end{aligned}
% \end{equation*}
% Rearranging,
% \begin{equation} \label{eq:k-means-alpha-solution-penalty-polynomial-2}
%     \sum_{j \in X^4} c_j \leq \frac{1}{1-2 \varepsilon} \left(\sum_{j \in X^1} c_j^*+(1+2 \varepsilon) \cdot \sum_{j \in X^4} c_j^*\right) \leq(1+6 \varepsilon) \cdot\left(\sum_{j \in X^1} c_j^*+\sum_{j \in X^4} c_j^*\right)
% \end{equation}
%Now we can proceed similarly to the proof of Lemma 2 in \cite{friggstad2019exact}. 
By Lemma~\ref{lem:c4}, we bound $\costp(S)$ in the following way:
\begin{equation} \label{eq:k-means-alpha-solution-penalty-polynomial-3}
\begin{aligned}
\costp(S) & \leq \costp(O_q)+2 \varepsilon \cdot \Psi_{pen}(S, O_q)=\costp(O_q)+2 \varepsilon \sum_{j \in X^4} c_j^*+2 \varepsilon \sum_{j \in X^4} c_j \\
& \leq \costp(O_q)+2 \varepsilon \sum_{j \in X^4} c_j^*+2 \varepsilon \cdot(1+6 \varepsilon) \cdot \left(\sum_{j \in X^1} c_j^*+\sum_{j \in X^4} c_j^* \right) \\
& \leq \sum_{j \in X^1}(1+6 \varepsilon) \cdot c_j^*+\sum_{j \in X^2} c_j^*+\sum_{j \in X^3} c_j^*+\sum_{j \in X^4}(1+6 \varepsilon) \cdot c_j^* \\
& \leq \sum_{j \in X^1}(1+6 \varepsilon) \cdot c_j^*+\sum_{j \in X^2} \min \left\{(1+6 \varepsilon) \cdot c_j^*, c_j\right\}+\sum_{j \in X^3} c_j^*+\sum_{j \in X^4}(1+6 \varepsilon) \cdot c_j^*.
\end{aligned}
\end{equation}
The last bound again uses $c_j^* \leq c_j$ for $j \in X^2$. Recall we chose $\varepsilon$ so that $(1+6 \varepsilon)=\left(1+\varepsilon^{\prime}\right)^2$. Thus, combining Lemma~\ref{lem:c4}, \eqref{eq:k-means-alpha-solution-penalty-polynomial-3} and the simple observation that $\costp^{\prime}(S)=\costp(S)$, we get that $\costp^{\prime}(S)=\costp(S) \leq \costp^{\prime}(O_q)$. 
% By the stability condition, $O_q$ is also an optimal solution for the perturbed instance, which
Which means that $S$ is an optimum solution for $(X, C, p^{\prime}, \delta^{\prime})$.
% (by using the stability condition one more time).
%, and since $\cI$ is $(1 + \varepsilon^{\prime})$-stable, this indicates that $S$ must be also an optimum solution for $\cI$. 
This completes the proof.

\section{Hardness of $(\alpha, \beta)$-stable \kmd} \label{sec:hard-alpha-beta}
Section~\ref{sec:pol} discussed steps towards finding optimum and near-optimum solutions of $(\alpha, \beta)$-stable \kms / \kmd with or without penalties. To complement that, in this section, we establish hardness results for $(\alpha, \beta)$-stable \kmd, as the question of whether we can efficiently obtain the exact optimum solution of $(\alpha, \beta)$-stable \kms / \kmd with penalties remains open. In Section~\ref{sec:hard-alpha-beta-penalty} and Section~\ref{sec:hard-alpha-beta-no-penalty}, we obtain hardness results for $(\alpha, \beta)$-stable \kmd, respectively, with and without penalties for any $\alpha\in(1,1.2)$, $\beta>0$. While, for ease of presentation, our hardness results are shown for \kmd, every result in this section readily extends to \kms as well.

The results of this section are inspired by a reduction of the \gti problem to \xmd{k^2} with penalties, which was introduced by \cite{cohen2018bane}. We here present our adapted version of the reduction from \gti to \xmd{k^2}.

Consider any \gti problem with integer $n$ and collection $S$ of $k^2$ non-empty sets. We define a Euclidean \xmd{k^2} instance with penalties similar to the one introduced in Theorem 6.2 of \citep{cohen2018bane}. Let $\cI(S, n) = (X^{grid}, C^{grid}, p^{grid}, \delta)$ denote this instance. Fix $\varepsilon= O(\beta / n^3)$ for any $\beta > 0$. Define a set of data points $X^{grid}$ in the region consisting of a square of side length $2 k+\varepsilon(n-1)$, where the lower left corner of the square is on the origin (i.e., $A=\{(x, y) \mid$ $0 \leq x, y \leq 2 k+\varepsilon(n-1)\})$. They are spaced evenly in a grid $G$, with two consecutive (horizontal or vertical) data points at distance $\varepsilon$ from each other; thus there are $\varSigma=(2 k / \varepsilon+n)^2$ data points. Each data point $\vec j$ has a penalty of $p^{grid}(\vec j) = 1$. Thinking of this grid as a discrete approximation of the uniform measure on the square $A$, we work with the discrete measure $\mu$ carried by the data points, where each data point is weighted 1, so that $\iint_A d \mu=\varSigma$.

For each set $S_{i, j}$, we introduce $\left|S_{i, j}\right| \leq n^2$ candidate centres, and let $C_{i, j}$ denote the set of such candidate centres (note that there are $k^2$ such sets), where $C_{i, j}=\{(2 i-1,2 j-1)+\varepsilon(u-1, v-1) \mid$ $\left.(u, v) \in S_{i, j}\right\}$. Note that the candidate centres are also placed on vertices of $G$, and that if $S_{i, j}$ has all possible pairs so that $S_{i, j}=[n] \times[n]$, then $C_{i, j}$ precisely forms a subgrid of $n^2$ evenly spaced points in which consecutive points are at distance $\varepsilon$ from each other and the lower left point of the subgrid lies at $(2 i-1,2 j-1)$. The final set of candidate centres is given by $C^{grid}=\cup_{1 \leq i, j \leq k} C_{i, j}$.

\begin{theorem} [Theorem 6.2 of \cite{cohen2018bane}] \label{thm:cohen-6-2}
There exist a $\nu \geq 0$, such that the \gti problem with integer $n$ and collection $S$ has a solution if and only if there exists a solution to $\cI(S, n)$  with cost at most $\nu$.
\end{theorem}

Previously, we mentioned that $\mu$ approximates the uniform measure. Let us elaborate on what we mean by this. Define, $R =[-0.5\varepsilon, 2k + \varepsilon (n - 0.5)] \times [-0.5\varepsilon, 2k + \varepsilon (n - 0.5)]$ to be a square with bottom left corner $(-0.5\varepsilon, -0.5\varepsilon)$ and top right corner $(2k + \varepsilon (n - 0.5), 2k + \varepsilon (n - 0.5))$. Denote by $u_R$ the uniform probability measure on $R$ and define $\mu^\star := \varSigma \cdot u_R$. In other words, $\mu^\star$ is a constant measure on $R$ with the property $\mu^\star(R) = \varSigma$ and $\mu^\star(M) = 1$ for any unit square $M$. Notice that for any measurable set $D$ and measurable function $f$, $\varepsilon^2 \iint f d \mu$ can be viewed as the two-dimensional Riemann sum for $\iint_{D} f d \mu^\star$. Therefore, $\iint_D f d \mu$ approximates $\frac{1}{\varepsilon^2} \iint_D f d \mu^\star$. The following Lemma will formalize this approximation for two specific $f$ and $D$ that will be used later. We defer the proof to the appendix.
      \begin{lemma} \label{lem:mu-approx}
     Fix any $\vec a \in C$. Define $D_r$ to be the circle with centre $\vec a$ and radius $r \leq 1$. Then \\
     (i)  \quad $\frac{1}{\varepsilon^2 }\iint_{D_{r - \varepsilon}} d \mu^\star  \leq \iint_{D_{r}} 1 d \mu \leq \frac{1}{\varepsilon^2 }\iint_{D_{r + \varepsilon}} d \mu^\star$, \\
     (ii) \quad $\iint_{D_r} \delta(\vec x, \vec a) d \mu \leq \frac{1}{\varepsilon^2} \left(\iint_{D_{r + \varepsilon}} \delta(\vec x, \vec a)  d \mu^\star +
         \varepsilon \iint_{D_{r + \varepsilon}} d \mu^\star \right).$\\
     Here $\iint_{D_{r}} 1 d \mu^\star = \pi r^2$ and $\iint_{D_r} \delta(\vec x, \vec a)  d \mu^\star = \frac{2}{3} \pi r^3$. 
         
 %     \begin{enumerate}
        
 %         \item $\frac{1}{\varepsilon^2 }\iint_{D_{r - \varepsilon}} d \mu^\star  \leq \iint_{D_{r}} 1 d \mu \leq \frac{1}{\varepsilon^2 }\iint_{D_{r + \varepsilon}} d \mu^\star$.
 %        % \item $\frac{1}{\varepsilon^2} \left(\iint_{D_{r - \varepsilon}} \delta(\vec x, \vec a)  d \mu^\star -
 %        %  \varepsilon \iint_{D_{r + \varepsilon}} d \mu^\star \right) \leq \iint_{D_r} \delta(\vec x, \vec a) d \mu \leq \frac{1}{\varepsilon^2} \left(\iint_{D_{r + \varepsilon}} \delta(\vec x, \vec a)  d \mu^\star +
 %        %  \varepsilon \iint_{D_{r + \varepsilon}} d \mu^\star \right)$
 %         \item \begin{tabular}{l l}
 %             $\frac{1}{\varepsilon^2} \left(\iint_{D_{r - \varepsilon}} \delta(\vec x, \vec a)  d \mu^\star -
 %         \varepsilon \iint_{D_{r + \varepsilon}} d \mu^\star \right)$ & $\leq \iint_{D_r} \delta(\vec x, \vec a) d \mu$ \\
 %         &  $\leq \frac{1}{\varepsilon^2} \left(\iint_{D_{r + \varepsilon}} \delta(\vec x, \vec a)  d \mu^\star +
 %         \varepsilon \iint_{D_{r + \varepsilon}} d \mu^\star \right). $
 %         \end{tabular}
 %     \end{enumerate}
 \end{lemma}

\subsection{Stable Euclidean \kmd with Penalties} \label{sec:hard-alpha-beta-penalty}

In this section, we %show the hardness of Euclidean $(\alpha, \beta)$-stable \kmd instances with penalties in $\bR^2$. We are going to 
will prove that for any $1 < \alpha < 1.2$, $0 < \beta$ 
 $(\alpha, \beta)$-stable Euclidean \kmd with penalties in $\bR^2$ is hard.
 The proof idea is to show that $\cI(S, n)$ is also stable. Since any two centres in any $C_{i, j}$ for $1 \leq i, j \leq k$ are close, for showing $(\alpha, \beta)$-stability it is enough to show that, even if we perturb distances by a constant factor, the optimum solution would be $k^2$ candidate centres each from a $C_{i, j}$ for $i, j \leq k$. Intuitively, this is not hard to see, as the cost of switching from $\vec{a} \in \mathcal{C}_{i, j}$ to another $\vec{b} \neq \vec{a} \in \mathcal{C}_{i, j}$ is relatively small. Conversely, by adding a centre $\vec{c}$ in a cluster $\mathcal{C}_{i', j'}$ that currently lacks a centre, we can significantly reduce the costs for points located near $\vec{c}$.

\kmedpenaltyhardness*

% \begin{theorem}\label{thm:k-median-penalty-alpha-beta-stable-hardness-mult-mult}
%     For any $1 < \alpha < 1.2$, $0 < \beta$, there is no $f(k) n^{o(\sqrt{k})}$-time algorithm for $(\alpha, \beta)$-stable Euclidean \kmd with penalties 
%     % w.r.t. $\CS^{met-pen}_{mult-mult}$ and $\CS^{sol}_{bij}$
%     in $\bR^2$, unless \ETH fails.
% \end{theorem}

\begin{proof}
 By Theorem~\ref{thm:cohen-6-2}, it suffices to show that $\cI(S, n)$ is also $(\alpha, \beta)$-stable. Consider any $\delta \leq \delta^\prime \leq \alpha \cdot \delta$ and $p \leq p^\prime \leq \alpha \cdot p$. Define $\costp^\prime(S) = \sum_{j \in X^{grid}} \min_{i \in S} \left\{\min \left\{\delta^\prime(i, j), p^\prime(j)\right\}\right\}$ for every $S \subseteq C^{grid}$. First we prove that every optimal solution $O$ of the \xmd{k^2}, $(X^{grid}, C^{grid}, \delta^{\prime}, p')$ cannot have two centres in a $C_{i, j}$ for $1 \leq i, j, \leq k$. For the sake of contradiction assume that there exist $1 \leq i, j, \leq k$ such that there exist $\vec a, \vec b \in C_{i, j} \cap O$. Due to the pigeonhole principle, we have $i'$ and $j'$, $1 \leq i^{\prime}, j^{\prime}, \leq k$ such that $C_{i^{\prime}, j^{\prime}} \cap O = \varnothing$. Let $O^{\prime} = (O \backslash \{\vec a\}) \cup \{\vec c\}$ for any $\vec c \in C_{i^{\prime}, j^{\prime}}$. We then argue that $\costp^\prime(O^{\prime})$ is smaller than $\costp^\prime(O)$, which is a contradiction. Let $D$ be the set of all points in a circle with centre $\vec c$ and radius $\frac{1}{\alpha}$. Every $\vec x \in D$ satisfies $\delta^{\prime}(\vec x, \vec c) \leq \alpha \delta(\vec x, \vec c) \leq 1$. Also, since every other point $\vec d \neq \vec a$, $\vec d \in O^\prime$, has the property $\delta(\vec c, \vec d) \geq 2 -  (n - 1) \varepsilon$, we have $\delta^\prime(\vec d, \vec x) \geq \delta(\vec d, \vec x) \geq 2 - \alpha - (n - 1) \varepsilon$ which is greater than 1 for large enough $n$ (for $n$ polynomial in $\frac{1}{\alpha - 1}$). Also, $p'(\vec x) \geq p^{grid}(\vec x) = 1$ which is smaller than $\delta^{\prime}(\vec x, \vec c)$. Therefore, every point in $D$ will be assigned to $\vec c$. Thus, adding $\vec c$ to $O$ will decrease $\costp^\prime(O)$ at least by
\begin{equation} \label{eq:k-median-penalty-alpha-beta-stable-hardness-1}
\begin{aligned}
    \iint_{D} (1 - \delta^{\prime}(\vec x, \vec c)) d\mu &\geq
    \iint_{D} \left(1 - \alpha \delta(\vec x, \vec c)\right) d\mu \\
    & \stackrel{(i)}{\geq} \frac{1}{\varepsilon^2} \left( \pi \left(\frac{1}{\alpha} - \varepsilon\right)^2  -  \pi \alpha \left( \frac{2}{3} \cdot \left(\frac{1}{\alpha} + \varepsilon\right)^3  + \varepsilon \left(\frac{1}{\alpha} + \varepsilon\right)^2 \right) \right),\\
    % & = \frac{1}{\varepsilon^2} \left( \frac{\pi (1 - \alpha)^2 (2 + \alpha)}{3} + O(n \varepsilon)   \right)
\end{aligned}
\end{equation}
where (i) is due to Lemma~\ref{lem:mu-approx}. Next, note that in order for any $\vec x \in A$ to be assigned to $\vec a$, $\delta^{\prime}(\vec x, \vec a)$ should be less than $p'(\vec x) \leq \alpha$. Subsequently, all such $\vec x$ should have $\delta(\vec x, \vec a) \leq \alpha$, i.e., they should belong to a circle with radius $\alpha$ and centre $\vec a$, which we denote by $D^{\prime}$. Next, notice that
\begin{equation} \label{eq:k-median-penalty-alpha-beta-stable-hardness-2}
    \begin{aligned}
\delta^{\prime}(\vec x, \vec b) & \leq \alpha \delta(\vec x, \vec b) \leq \alpha (\delta(\vec x, \vec a) + \varepsilon (n - 1))  = \delta(\vec x, \vec a) + \left(\alpha - 1\right) \delta(\vec x, \vec a) + \alpha \varepsilon (n - 1) \\
& \leq \delta^{\prime}(\vec x, \vec a) + \left(\alpha - 1\right) \delta(\vec x, \vec a) + \alpha \varepsilon (n - 1). 
\end{aligned}
\end{equation}
Thus, removing $\vec a$ from $O$ will decrease $\costp^\prime(O)$ by at most
\begin{equation} \label{eq:k-median-penalty-alpha-beta-stable-hardness-3}
    \begin{aligned}
        \iint_{D^{\prime}} &\max \left[ \delta^{\prime}(\vec x, b) - \delta^{\prime}(\vec x, \vec a), 0\right] \stackrel{(i)}{\leq} 
     \iint_{D^{\prime}} \left(\alpha - 1\right) \delta(\vec x, a) + \alpha \varepsilon (n-1) d\mu \\
    &\leq \frac{1}{\varepsilon^2} \left( \left(\alpha - 1\right) \frac{2\pi}{3} (\alpha + \varepsilon)^3 + (\alpha - 1)\varepsilon \pi (\alpha + \varepsilon)^2 + \alpha \varepsilon (n-1) \pi ( \alpha + \varepsilon)^2  \right). \\
    % & = \frac{1}{\varepsilon^2} \left(\frac{\pi}{3} + O(n \varepsilon) \right)
    \end{aligned}
\end{equation}
Next we combine \eqref{eq:k-median-penalty-alpha-beta-stable-hardness-1} and \eqref{eq:k-median-penalty-alpha-beta-stable-hardness-3} to derive $\costp^\prime(O^{\prime}) - \costp^\prime(O) \leq \frac{1}{\varepsilon^2} \left((\alpha^4 - \alpha^3) \frac{2 \pi}{3} - \frac{ \pi}{3 \alpha^2}  + O(n \varepsilon) \right)$,
% $$
% \begin{aligned}
%     \costp^\prime(O^{\prime}) - \costp^\prime(O) &\leq \frac{1}{\varepsilon^2} \left((\alpha^4 - \alpha^3) \frac{2 \pi}{3} - \frac{ \pi}{3 \alpha^2}  + O(n \varepsilon) \right) \\
% \end{aligned}
% $$ 
which is less than zero for $1 < \alpha < 1.2$ for large enough $n$. This  contradicts $O$ minimizing $\costp^\prime$.

% Next, let $\overline{O}$ be the set of optimal solutions of $\cI(S, n)$. Up until here we established that $O$ contains $k^2$ centres, one in each $C_{i, j}$ for $1 \leq i, j, \leq k$. Since $(\delta, P) \in \CS^{met-pen}_{mult-mult}(\delta, p^{grid})$, every $O^\prime \in \overline{O}$ has the same property. For any $O^\prime \in \overline{O}$ let $f:O^\prime \rightarrow O$ be the unique function such that for any $o \in O^\prime$, if $o \in C_{i, j}$ for some $1 \leq i, j, \leq k$, then we'd also have $f(o) \in C_{i, j}$. Therefore, 

Consider $O^\prime \in \cF_{k^2}$ to be any optimum solution of $\cI(S, n)$. So far, we have established that $O$ contains $k^2$ centres, one in each $C_{i, j}$ for $1 \leq i, j, \leq k$. This indicates that $O^\prime$ has the same property. Let $f:O^\prime \rightarrow O$ be the unique function such that for any $o \in O^\prime$, if $o \in C_{i, j}$ for some $1 \leq i, j, \leq k$, then also $f(o) \in C_{i, j}$. Therefore,  
$\sum\nolimits_{o \in O_p} \delta(o, f_p(o)) \leq k^2 \varepsilon (n-1) \leq \beta$.
This implies that 
% $O \in \CS^{sol}_{bij} (\overline{O}, \beta)$
$\distbij{\cI}(O, O^\prime) \leq \beta$. Thus, $\cI(S, n)$ is $(\alpha, \beta)$-stable.
% w.r.t. $\CS^{met-pen}_{mult-mult}$ and $\CS^{sol}_{bij}$.
\end{proof}

\subsection{Stable \kmd in Doubling Metrics} \label{sec:hard-alpha-beta-no-penalty}
Focusing on the setting without penalties, we prove that
 $(\alpha, \beta)$-stable \kmd 
 % w.r.t. $\CS^{met-pen}_{mult-mult}$ and $\CS^{sol}_{bij}$ 
 with doubling dimension $3$ is hard, for any $1 < \alpha < 1.2$, $0 < \beta$. The idea for the following theorem is that while the point set is located in $\bR^3$, introducing a metric $\delta^{pen}$ with doubling dimension $3$ (different from standard norm two metric) allows us to replicate the effect of penalties. The proof of the following theorem is deferred to the appendix.

\kmedhardness*

\section{Hardness of $\alpha$-stable \kmd in Bounded-Dimensional Metrics} 
\label{sec:hard-alpha}
% In this section we show that solving $1 + \frac{1}{poly(n)}$-stable \kms is \wonehard.

%\cite{friggstad2019exact} showed that for any constant $\varepsilon$, $(1 + \varepsilon)$-stable instances of \kmd are solvable in polynomial time when the doubling dimension is a constant. They also showed that having a bounded-dimensional space is crucial as solving the problem in cases where the dimension is part of the input would imply  $\text{NP} = \text{RP}$. Indeed, resolving the complexity of perturbation-resilient \kmd and \kms for constant values of $\varepsilon$ 

\cite{friggstad2019exact} set out to resolve the complexity of $(1 + \varepsilon)$-stable instances of \kmd and \kms for constant values of $\varepsilon$, which is highly desirable since this is the range where the $\rho$-swap local search can produce solutions in polynomial time. However, the computational complexity of the problem remains unknown when we go beyond constant error values. This question is of particular interest when discussing Fixed Parameter Tractable (\FPT) algorithms that allow us slightly super-polynomial running times. Therefore, we examine the computational complexity of $(1 + \varepsilon)$-stable instances of \kmd (and \kms) when $\varepsilon$ is an inverse polynomial in $n$. As the main result of this section, we show that for $\varepsilon = \frac{1}{poly(n)}$, the problem is still hard. 

Inspired by \cite{cohen2018bane}, our general approach is to reduce \pvc to $\kmd$ using moment curves. We first establish some useful properties regarding moment curves and spheres. These properties resemble those shown in Section 2.1 of \cite{cohen2018bane}, but are modified to our purpose. The proofs, in particular, utilize Descartes’ rule of signs (see \cite{curtiss1918recent}), and are deferred to the appendix. 
% Lemma~\ref{lem:mom-curve} and Lemma~\ref{lem:mom-curve-2} are modified versions of Lemma 2.7 and Lemma 2.8 from \cite{cohen2018bane}

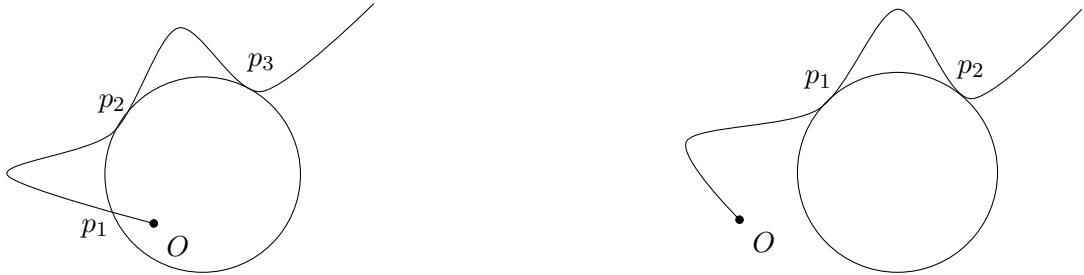
\begin{figure*}[ht!]
\begin{tikzpicture}[scale=0.65]
\begin{scope}[
  vertex/.style={
    draw,
    circle, fill,
    minimum size=1mm,
    inner sep=0pt,
    outer sep=0pt%
  },
]

% Center of the circle
\node[vertex, label={below right:$O$}] (o) at (0,0) {};

% Points
\node[label={$p_1$}] at (-1.2,-0.7) {};
\node[label={$p_2$}] at (-0.85,1.85) {};  % Merging p2 and p3
\node[label={$p_3$}] at (2.2,2.7) {};  % Merging p4 and p5

% Curved path with smooth interpolation
\draw plot [smooth] coordinates {(0,0) (-3,1) (-0.85,1.85) (0.5,4) (2.2,2.7) (4.5,4.5)}; 

% Circle centered at (1,1) with radius 2
\draw (1,1) circle (2cm);

\end{scope}
\end{tikzpicture}
\hspace{0.25\linewidth}
\begin{tikzpicture}[scale = 0.7]
    \begin{scope}[
        vertex/.style={
            draw,
            circle, fill,
            minimum size=1mm,
            inner sep=0pt,
            outer sep=0pt%
        },
        ]
        
        \node[vertex,label={below right:$O$}] (o) at (0,0) {};
        % \node[label={$p_1$}] (o) at (0.95,1.15) {};
		\node[label={$p_1$}] (o) at (1.5,2.1) {};
		\node[label={$p_2$}] (o) at (4.4,2.3) {};
		% \node[label={$p_2$}] (o) at (5,1.45) {};
        \draw plot [smooth] coordinates {(0,0) (-1,1.5) (1.5, 2.1) (3,4) (4.4, 2.3) (6.5,4)}; 
        \draw (3,0.9) circle (1.9cm);
    \end{scope}
\end{tikzpicture}
%\label{fig:curve}
    \caption{In $\bR^4$ (left), the $3$-sphere that goes through the points $p_1, p_2, p_3$ on the moment curve and is tangent to the moment curve on $p_2$ and $p_3$, has no other intersections with the moment curve
after the origin. 
        In $\bR^3$ (right), the sphere that is tangent to the moment curve on $p_1$ and $p_2$ has no other intersections with the moment curve.
\label{fig:sphereandcurve}}
\end{figure*}

\begin{lemma} \label{lem:mom-curve}
      Fix any 3 positive values $0<t_1<t_2<t_3$. Consider the 4-dimensional moment curve $\left(t, t^2, t^3, t^4\right)$ and the 3-sphere in $\bR^4$ that goes through the moment curve at $t = t_1$ and is tangent to the curve at $t = t_2, t_3$. The segments on the moment curve corresponding to $t \in\left(t_1, t_2\right) \cup\left(t_2, t_3\right) \cup (t_3, \infty)$ all lie outside of the sphere (i.e., the distance of all such points from the centre of the 3-sphere is strictly more than its radius).
\end{lemma}

\begin{lemma} \label{lem:mom-curve-2}
      Fix any 2 positive values $0<t_1<t_2$. Consider the 3-dimensional moment curve $\left(t, t^2, t^3\right)$ and the unique 2-sphere in $\bR^3$ that is tangent to the moment curve at $t = t_1, t_2$. The sphere only contacts the curve at $t_1$ and $t_2$.
      % \sanr{I prefer the wording in the previous lemma, saying that for certain intervals, teh curve lies outside the sphere. The ``only contacts'' sounds ambiguous.}
\end{lemma}

\subsection{$(1 + \frac{1}{poly(n)})$-Stable Euclidean \kmd with Penalties}% in Four Dimensions}
In this subsection, we focus on $\kmd$ with penalties and show that solving $(1 + \frac{1}{poly(n)})$-stable Euclidean \kmd with penalties 
% w.r.t. $\CS^{met-pen}_{mult-mult}$ 
in $\bR^4$ is hard. The structure of the proof is as follows. Just like \cite{cohen2018bane}, we reduce from \pvclong (\pvc) by creating a candidate centre on a moment curve for each vertex of the input graph, along with a data point for each edge of the input graph $G = (V, E)$. For each edge, we ensure that the point corresponding to that edge is closer to the two centres representing each vertex of the edge than to any other candidate centre. This allows us to reduce a \pvc instance to a $\kmd$ instance.

However, to ensure the stability of $\kmd$ with penalties, we need to guarantee that the difference between the distances of the representation of an edge to its corresponding vertex and to the representation of other vertices is significant. Moreover, with the reduction introduced by \cite{cohen2018bane}, perturbing the metric function risks having the optimum solution switch from one maximal partial cover of the graph to another. To avoid this, we ensure that the distance between the representation of edges and the representation of their respective vertices remains constant. These are the key properties in our proof, which is deferred to the appendix.

% Lemma~\ref{lem:mom-curve-2} allows us to define $n$ points on a moment curve and for each duo of them there is another point (centre of the sphere defined) that has closer distance to them than any other points. We use this property to represent the vertices and edges of a graph with respectively points on the curve and and the centres of the spheres.

\kmedpeninversehardness*

\subsection{Hardness of $(1 + \frac{1}{poly(n)})$-Stable Euclidean \kmd}
Next, we focus on $\kmd$ without penalties and show that $(1 + \frac{1}{poly(n)})$-stable Euclidean \kmd 
% w.r.t. $\CS^{met-pen}_{mult-mult}$ 
in $\bR^6$ is hard. The proof is similar to that of Theorem~\ref{thm:k-median-alpha-hardness-panlty}, and is deferred to the appendix. %However, inspired by \cite{cohen2018bane} we introduce a point with the role of replicating penalty, which makes the proofs more complicated.

\kmedinversehardness*

\section{Conclusion}

This paper addresses the complexity of stable instances of \kms and \kmd under generalized notions of stability, and in the generalized setting with penalties. We  show that, under our most general definition of stability, i.e., $(\alpha, \beta)$-stability, the problem appears to be highly intractable, even in small dimensional Euclidean spaces. We fell short of answering the following question.

\smallskip

\noindent \textbf{Open Problem:} Does a $(1 + \varepsilon)$-approximation for any constant $\varepsilon > 0$ exist for almost stable (i.e., $\varepsilon_\circ, \beta_\circ)$-stable) instances of \kms/\kmd?

Note that in the almost-stable setting, a $(1 + \varepsilon)$-approximation would mean to find a solution that is both $(1 + \varepsilon)$-close in cost and $\beta$-close in solution structure to an optimum solution of the given instance. This second requirement implies that, for example, the known PTASs for the \kmd \cite{friggstad2019local,C-AKM19} do not necessarily produce the desired solution in this setting.
\newpage 

\bibliography{ref}

\appendix

% \crefalias{section}{appendix} % uncomment if you are using cleveref

\section{Omitted Proofs of Section~\ref{sec:pol}}

\subsection{Proof of Lemma~\ref{lem:rho-swap-terminate}}
\farrd{changed this proof}
\begin{proof}
    To establish this result, our analysis requires extra care compared to \cite{friggstad2019exact}. Friggstad et al.\ introduce the notion of being nearly-good with respect to a fixed reference point, namely the unique optimal solution. In our setting, we generalize stability to accommodate instances with multiple optimal solutions. As a result, the absence of a single fixed reference point introduces new complications, which we discuss below.
    
    We first argue that the algorithm terminates at a nearly-good solution. According to the definition of nearly-good solutions, any optimum solution could cause the solution $S$ to violate the nearly-good solution condition. We say that $S$ is nearly-good for $F$ if $cost(S) - cost(F) \leq 2\varepsilon \Psi(S, F)$. Algorithm~\ref{alg:rho-swap} clearly terminates as in each round it selects a new $S$ with strictly smaller cost. Assume that the set $S$ in the last round is not nearly-good, meaning that there exists an \far{optimum} $F \in \cF_k$ for which $S$ is not nearly-good. Observe that by Theorem \ref{thm:cost-drop}, there should exist $S' \in \cF_k$ such that $\left|S-S^{\prime}\right| \leq \rho$ and $\cost\left(S^{\prime}\right) \leq \cost(S)+\frac{\cost(F)-\cost(S)+\varepsilon \cdot \Psi(S, F)}{k} < \cost(S)$, since the numerator of the fraction is negative. This is a contradiction because the $\rho$-swap local search could discover $S'$ and move to it in the next iteration, contradicting the assumption that $S$ was the returned solution.
    
    It only remains to prove a nearly-good solution is found within $2 k \cdot \ln (n \Delta)$ iterations. Let $\cO = \{O_1, O_2, \ldots, O_\ell\}$ be the set of optimal solutions. For the sake of contradiction, suppose that after $K=[2 k \cdot \ln (n \Delta)]$ iterations Algorithm 1 has still not encountered a nearly-good solution. Say $S_0, S_1, \ldots, S_K \in \mathcal{F}_k$ is the sequence of sets held by the algorithm after the first $K$ iterations, where $S_0$ is the initial set. Then for all $i = 0, ..., K$ we should have $S_i$ is not nearly-good at least with some $O_{i} \in \cO$. Therefore, by Theorem \ref{thm:cost-drop}, for each $i \in [K]$ we have
    
    \begin{align*}
        \cost\left(S_i\right)-\cost(O_i)   & \leq \cost(S_{i - 1})-\cost(O_i)+\frac{\cost(O_i)-\cost(S_{i - 1})+\varepsilon \Psi(S_{i - 1}, O_i)}{k} \\
                                                & < \cost(S_{i - 1})-\cost(O_i)+\frac{\cost(O_i)-\cost(S_{i - 1})}{2 k} \\
                                                & = \left(1-\frac{1}{2 k}\right) \cdot(\cost(S_{i - 1})-\cost(O_i)) \\
                                                & = \left(1-\frac{1}{2 k}\right) \cdot(\cost(S_{i - 1})-\opt) \\
    \end{align*}
    The second inequality follows from the fact that $\varepsilon \Psi(S_{i - 1}, O_i) < \frac{1}{2} \Psi(S_{i - 1}, O_i) \leq \frac{1}{2}(\cost(S_{i - 1}) + \cost(O_i))$ by our choice of $\varepsilon$. Because costs are integral, $\cost(S_K) - \opt = 0$ which contradicts that $S_K$ is not a nearly-good solution
\end{proof}

\subsection{Proof of Theorem~\ref{thm:cost-drop-penlty}}

\begin{proof}
 Sample a random partition $\pi$ of $O^{\prime} \cup S^{\prime}$ with mentioned properties and consider the effect of the swap $S \rightarrow S \triangle T$ for each part $T \in \pi$. We bound $\mathbf{E}_\pi\left[\sum_{T \in \pi} \costp(S \triangle T)-\costp(S)\right]$ from above by describing a valid way to redirect each $j \in X$ in each swap. We use the shorthand $c_j^*=\min(\delta(j, \sigma(j, O))^2, p(j))$ for the cost of connecting $j$ in $O$ and  $c_j=\min(\delta(j, \sigma(j, S))^2, p(j))$ similarly for $S$. Based on possible reassignments of $j \in X$, we break the analysis into four cases:
 
 \textbf{Case 1:} Either $\{\sigma(j, S), \sigma(j, O)\} \subseteq S \cap O$ or both $\delta(j, O)^2, \delta(j, S)^2 \geq p(j)$. The former means that $\sigma(j, S)$ remains open after each swap, and not reassigning $j$ is a valid option; the latter means paying the penalty for $j$ is the better option. Thus,  the reassignment cost is bounded by $c_j^* - c_j = 0$.
     % Note $\sigma(j, S)$ remains open after each swap so this is valid.
     %Observe for such clients that $c_j^*=c_j$ so we, conveniently, say the total assignment cost change for $j$ over all swaps $T \in \pi$ is bounded by $c_j^*-c_j$. 
 
     \textbf{Case 2:} $\sigma(j, S) \in S^{\prime}$ and $\sigma(j, O) \in S \cap O$ and $\delta(j, S)^2 < p(j)$ (note that if the last condition does not hold, we are in the previous case). Move $j$ to $\sigma(j, O)$ when swapping the part $T$ with $\sigma(j, S) \in T$. As $\sigma(j, S)$ remains open when swapping all other $T^{\prime} \neq T$, we can leave $j$ assigned to $\sigma(j, S)$ to upper-bound its cost change for swaps $T^{\prime} \neq T$ by 0. The total cost assignment for $j$ is then bounded by $c_j^*-c_j$.
     % Since either $\sigma(j, O)$ is open or $\delta(j, O)^2 \geq p(j)$, for every $p$ the cost assignment for $j$ is bounded by $c^*_j - c^j$.

  \textbf{Case 3:} $j$ with $\sigma(j, S) \in S \cap O$ and $\sigma(j, O) \in O^{\prime}$ and $\delta(j, O)^2 < p(j)$. Move $j$ to $\sigma(j, O)$ when swapping the part $T$ with $\sigma(j, O) \in T$ and do not move $j$ when swapping any other part $T^{\prime} \neq T$. This places an upper bound of $c_j^*-c_j$ on the total assignment cost change for $j$.
 
     \textbf{Case 4:} Finally, consider $j$ with $\sigma(j, S) \in S^{\prime}$ and $\sigma(j, O) \in O^{\prime}$ and $\delta(j, O)^2, \delta(j, S)^2 < p(j)$ (note that if the last condition doesn't hold, we are in one of the previous cases). Note these are precisely the points $j \in \overline{X}^{pen}_{S, O}$. From \eqref{eq:rand-partition-penalty},
$$
\mathbf{E}_\pi\left[\sum\nolimits_{T \in \pi} \Tilde{\Delta}_j^T\right] \leq(1+\varepsilon) \cdot c_j^*-(1-\varepsilon) \cdot c_j=c_j^*-c_j+\varepsilon \cdot\left(c_j+c_j^*\right)
$$
Aggregating this cost bound for all clients, we see
$$
\mathbf{E}_\pi\left[\sum\nolimits_{T \in \pi} \costp(S \triangle T)-\costp(S)\right] \leq \costp(O)-\costp(S) +\varepsilon \cdot \Psi_{pen}(S, O)
$$
Therefore there is some $\pi$ and some $T \in \pi$ with
$$
\begin{aligned}
    \costp(S \triangle T)-\costp(S) &\leq \frac{\costp(O)-\costp(S)+\varepsilon \cdot \Psi_{pen}(S, O)}{|\pi|} \\
    &\leq \frac{\costp(O)-\costp(S)+\varepsilon \cdot \Psi_{pen}(S, O)}{k}
\end{aligned}
$$
The last inequality relies on $|\pi| \leq k$ and $\costp(O)-\costp(S)+\varepsilon \cdot \Psi_{pen}(S, O) < 0$.
\end{proof}

\subsection{Proof of Lemma~\ref{lem:c4}}
\begin{proof}
The proof sis similar to that of \cite{friggstad2019exact}. As $S$ is a nearly-good solution, $c_j^* \leq c_j$ for $j \in X^2$, and $c_j^*=c_j$ for $j \in X^3$, we have:
\begin{equation*}
\begin{aligned}
\sum\nolimits_{j \in X^4} c_j & \leq \sum\nolimits_{j \in X^1} c_j+\sum\nolimits_{j \in X^4} c_j=\costp(S)-\sum\nolimits_{j \in X^2} c_j-\sum\nolimits_{j \in X^3} c_j \\
& \leq \costp(S)-\sum\nolimits_{j \in X^2} c_j^*-\sum_{j \in X^3} c_j^* \\
& \leq \costp(O_q)+2 \varepsilon \cdot \Psi_{pen}(S, O_q)-\sum\nolimits_{j \in X^2} c_j^*-\sum\nolimits_{j \in X^3} c_j^* \\
& =\sum\nolimits_{j \in X^1} c_j^*+\sum_{j \in X^4} c_j^*+2 \varepsilon\left(\sum\nolimits_{j \in X^4} c_j^*+c_j\right)
\end{aligned}
\end{equation*}
Rearranging,
\begin{equation} \label{eq:k-means-alpha-solution-penalty-polynomial-2}
    \sum_{j \in X^4} c_j \leq \frac{1}{1-2 \varepsilon} \left(\sum_{j \in X^1} c_j^*+(1+2 \varepsilon) \cdot \sum_{j \in X^4} c_j^*\right) \leq(1+6 \varepsilon) \cdot\left(\sum_{j \in X^1} c_j^*+\sum_{j \in X^4} c_j^*\right)
\end{equation}
\end{proof}

\section{Tools for Reduction Proofs}

We frequently use the moment curve throughout our reductions, which we define as follows.

\begin{definition}
    The curve $\bR^{+} \rightarrow \bR^d$ defined by $t \mapsto\left(t, t^2, \ldots, t^d\right)$ is called the moment curve.
\end{definition}

All of our lower bounds are conditioned on the Exponential Time Hypothesis (\ETH), which was formulated in \cite{IP99}.

\begin{definition}[Exponential Time Hypothesis (\ETH) \cite{}] There exists a positive real value $s>0$ such that 3-CNF-SAT, parameterized by $n$, has no $2^{sn}(n+m)^{O(1)}$-time algorithm (where $n$ denotes the number of variables and $m$ denotes the number of clauses).
\end{definition}

The following problem, \pvclong, plays a critical role in our reductions. In particular, for showing hardness of inverse poly of $n$ stable $\kmd$ (Section~\ref{sec:hard-alpha}).

\begin{definition}[\pvclong (\pvc)]
\textbf{Input:} A graph $G=(V, E)$, an integer $s \in \mathbb{N}$. \\
\textbf{Parameter:} Integer $k$. \\
\textbf{Output:} YES if and only if there exists a set of $k$ vertices that covers at least $s$ edges.
\end{definition} \farrd{\pvc used in the previous section}

The following lower bound is already shown for \pvc \cite{impagliazzo2001problems}.

\begin{theorem} [\pvc Hardness \cite{impagliazzo2001problems}]. There is no $f(k) n^{o(k)}$-time algorithm for the \pvclong problem unless \ETH fails (for any computable function $f$ ), where $n$ is the size of the input.
\end{theorem}

\gti problem also plays a critical role in our reductions. In particular, for showing hardness of $(\alpha, \beta)$-stability of $\kmd$ (Section~\ref{sec:hard-alpha-beta}).

\begin{definition} [\gti]
\textbf{Input:} Integer $n$, collection $S$ of $k^2$ nonempty sets $S_{i, j} \subseteq[n] \times[n]$ (where $1 \leq i, j \leq k$ ). \\
\textbf{Parameter:} Integer $k$. \\
\textbf{Output:} YES if and only if there exists a set of $k^2$ pairs $s_{i, j} \in S_{i, j}$ such that
\begin{itemize}
    \item  If $s_{i, j}=(a, b)$ and $s_{i+1, j}=\left(a^{\prime}, b^{\prime}\right)$, then $a \leq a^{\prime}$.
     \item If $s_{i, j}=(a, b)$ and $s_{i, j+1}=\left(a^{\prime}, b^{\prime}\right)$, then $b \leq b^{\prime}$.
\end{itemize}

\end{definition}

We call this problem \gti, and it is also known that this problem has no $f(k)n
^{o(k)}$-time algorithm unless \ETH fails \cite{cygan2015parameterized}.

\section{Omitted Proofs of Section~{\ref{sec:hard-alpha-beta}}}

\subsection{Proof of Lemma~\ref{lem:mu-approx}}

\begin{proof}

\emph{Claim 1.} For every $\vec y$ in $A$ define $M_{\vec y}$ to be square centred at $y$ with side length $\varepsilon$. Notice that every point in $R$ is in an $M_{\vec y}$ for some $\vec y \in A$, and since every point in each square is at most $\varepsilon$ away from its centre, we immediately derive $D_{r - \varepsilon} \subseteq \bigcup_{\vec y \in A \cap D_r} M_{\vec y} \subseteq D_{r + \varepsilon}$. Therefore,
    $$
    \iint_{D_{r - \varepsilon}} d \mu^\star \leq \iint_{\bigcup_{\vec y \in A \cap D_r} M_{\vec y}} d \mu^\star = \varepsilon^2 |A \cap D_r| = \varepsilon^2 \iint_{D_r} d \mu$$
    This yields the first inequality, and the second inequality can be attained in a similar manner.
    
    \emph{Claim 2.} 
 % For the first inequality, again using the fact that $D_{r - \varepsilon} \subseteq \bigcup_{\vec y \in A \cap D_r} M_{\vec y}$ and metric properties of $\delta$, we derive
 %    $$
 %    \begin{aligned}
 %        \iint_{D_{r - \varepsilon}} \delta(\vec a, \vec x) d \mu^\star &\leq \iint_{\bigcup_{\vec y \in A \cap D_r} M_{\vec y}} \delta(\vec a, \vec x)  d \mu^\star \\
 %        % &= \sum_{\vec y \in A \cap D_r}\iint_{M_{\vec y}} \delta(\vec a, \vec x)  d \mu^\star \\
 %        & \leq \iint_{\bigcup_{\vec y \in A \cap D_r} M_{\vec y}} \delta(\vec a, \vec y) + \delta(\vec x, \vec y)  d \mu^\star  \\
 %        & \iint_{\bigcup_{\vec y \in A \cap D_r} M_{\vec y}} \delta(\vec a, \vec y) d \mu^\star + \varepsilon 
 % \iint_{\bigcup_{\vec y \in A \cap D_r} M_{\vec y}} d \mu^\star  \\
 %         & = \sum_{\vec y \in A \cap D_r}\iint_{M_{\vec y}} \delta(\vec a, \vec y) d \mu^\star +   \varepsilon \iint_{\bigcup_{\vec y \in A \cap D_r} M_{\vec y}}  d \mu^\star \\
 %         & \stackrel{(i)}{\leq}  \varepsilon^2 \sum_{\vec y \in A \cap D_r} \delta(\vec a, \vec y) + \varepsilon \iint_{D_{r + \varepsilon}}  d \mu^\star \\ 
 %         & = \varepsilon^2 \iint_{D_r} \delta(\vec a, \vec x) d \mu + \varepsilon \iint_{D_{r + \varepsilon}} d \mu^\star \\
 %    \end{aligned}    
 %    $$
 %    Where (i) is due the fact that $\bigcup_{\vec y \in A \cap D_r} M_{\vec y} \subseteq D_{r + \varepsilon}$. 
    % For the second part,
    Again, using the fact that $\bigcup_{\vec y \in A \cap D_{r}} M_{\vec y} \subseteq D_{r + \varepsilon}$ we derive
    $$
    \begin{aligned}
        \iint_{D_{r + \varepsilon}} \delta(\vec a, \vec x) d \mu^\star &\geq \iint_{\bigcup_{\vec y \in A \cap D_{r }} M_{\vec y}} \delta(\vec a, \vec x)  d \mu^\star \\
        % &= \sum_{\vec y \in A \cap D_r}\iint_{M_{\vec y}} \delta(\vec a, \vec x)  d \mu^\star \\
        & \geq \iint_{\bigcup_{\vec y \in A \cap D_r} M_{\vec y}} \delta(\vec a, \vec y) - \delta(\vec x, \vec y)  d \mu^\star  \\
        & \geq \iint_{\bigcup_{\vec y \in A \cap D_r} M_{\vec y}} \delta(\vec a, \vec y) - \varepsilon 
 \iint_{\bigcup_{\vec y \in A \cap D_r} M_{\vec y}} d \mu^\star  \\
         & = \sum_{\vec y \in A \cap D_r}\iint_{M_{\vec y}} \delta(\vec a, \vec y) -   \varepsilon \iint_{\bigcup_{\vec y \in A \cap D_r} M_{\vec y}}  d \mu^\star \\
         & \stackrel{(i)}{\geq}  \varepsilon^2 \sum_{\vec y \in A \cap D_r} \delta(\vec a, \vec y) - \varepsilon \iint_{D_{r + \varepsilon}} d \mu^\star \\ 
         & = \varepsilon^2 \iint_{D_r} \delta(\vec a, \vec x) d \mu - \varepsilon \iint_{D_{r + \varepsilon}}  d \mu^\star \\
    \end{aligned}    
    $$
    Where (i) is due the fact that $\bigcup_{\vec y \in A \cap D_r} M_{\vec y} \subseteq D_{r + \varepsilon}$. 

     Finally, note that $\iint_{D_{r}} 1 d \mu^\star$ is simply the area of a circle with radius $r$, which is $\pi r^2$. Moreover, $\iint_{D_r} \delta(\vec x, \vec a)  d \mu^\star$ is the volume of a cylinder with radius and height $r$, which has a cone with radius and height $r$, carved out from it. Therefore, $\iint_{D_r} \delta(\vec x, \vec a)  d \mu^\star = \pi r^3 - \frac{1}{3} \pi r^3 = \frac{2}{3} \pi r^3$. This completes the proof.
\end{proof}

\subsection{Proof of Theorem~\ref{thm:k-median-alpha-beta-stable-hardness-mult-mult}}
\begin{proof}
   Define
   $\delta^{pen} (\vec x, \vec y) := \max \left(\|(x_1, x_2) - (y_1, y_2)\|_2, |x_3 - y_3|\right)$ for $\vec x = (x_1, x_2, x_3)$ and $\vec y = (y_1, y_2, y_3)$.
   Note that this is a metric, because
   $$
   \begin{aligned}
       \delta^{pen}(\vec x, \vec z) &= \max \left(\|(x_1, x_1) - (z_1, z_2)\|_2, |x_3 - z_3|\right) \\
       &  \leq \max \left(\|(x_1, x_1) - (y_1, y_2)\|_2 + \|(x_1, x_1) - (z_1, z_2)\|_2, |x_3 - y_3| + |y_3 - z_3|\right) \\
       &   \leq \max\left(\|(x_1, x_2) - (y_1, y_2)\|_2, |x_3 - y_3|\right) +  \max \left(\|(y_1, y_1) - (z_1, z_2)\|_2, |y_3 - z_3|\right) \\
       & \leq \delta^{pen}(\vec x, \vec y) + \delta^{pen}(\vec y, \vec z) 
   \end{aligned}
   $$
   Moreover, $\delta^{pen}$ has doubling dimension $3 = \log_2 8$, because any ball of radius $r$ w.r.t. $\delta^{pen}$ is a cylinder with radius and height $r$, which can be covered by $8$ balls of radius $r/2$.

   Define $X', C'$ to extend $X^{grid}, C^{grid}$ into 3 dimensions by setting their third coordinate to zero: 
   $$
    X' := \{ (u, v, 0) \mid (u, v) \in X^{grid} \}, \quad  C' := \{ (u, v, 0) \mid (u, v) \in C^{grid} \}
   $$
   Then define $C'' = \bigcup_{1 \leq i, j \leq k} \{ (2i - 1, 2j, 1),  (2i, 2j - 1, 1)$\} and $C := C' \cup C''$. Also, define $X$ to be $X'$ together with $\varSigma$ points in each of $(2i - 1, 2j, 1)$ and $(2i, 2j - 1, 1) $ for $1 \leq i, j \leq k$.

   We prove that solving the \xmd{3k^2} instance $(X, C, \delta^{pen})$ is equivalent to solving the \xmd{k^2} instance $(X^{grid}, C^{grid}, \delta, p^{grid})$. First, we show all $2k^2$ candidate centres in $C''$ should be selected in all optimum solutions of $(X, C, \delta^{pen})$. Note that every third coordinate of every candidate centre in $C''$ is 1. Thus, every data point in $X'$ has distance at least 1 to them. Moreover, note that for every $(u, w) \in X^{grid}$ such that $\|(u, w) , (2i, 2j - 1)\|_2 \leq 1$, we have $\delta^{pen}((u, w, 0), (2i, 2j - 1, 1)) = 1$. Similarly, for every $(u, w) \in X^{grid}$ such that $\|(u, w) , (2i - 1, 2j)\|_2 \leq 1$, we have $\delta^{pen}((u, w, 0), (2i, 2j - 1, 1)) = 1$. In other words, the circle with radius 1 around every $(2i, 2j - 1, 0)$ and $(2i, 2j - 1, 0)$ is distance 1 away from some centre in $C''$. These circles cover all of $X'$. Thus, the cost of any solution containing $C''$ will be at most $\varSigma$. On the other hand, for any solution that doesn't contain any $\vec a \in C''$, the cost of data points located in $\vec a$  alone is $\varSigma$.

  Next, since aforementioned circles cover $X'$, any data point with a distance of more than 1 from some selected candidate centre in $C'$ will be assigned to some centre in $C''$ and will have a cost of 1. This means that solving 
   \xmd{k^2} for $(X, C, \delta^{pen})$ is equivalent to solving \xmd{3k^2} for $(X^{grid}, C^{grid}, \delta, p^{grid})$. Finally, with an argument identical to our argument in Theorem~\ref{thm:k-median-penalty-alpha-beta-stable-hardness-mult-mult} we can also prove that $(X, C, \delta^{pen})$ is $(\alpha, \beta)$-stable 
   % w.r.t. $\CS^{met-pen}_{mult-mult}$ and $\CS^{sol}_{bij}$ 
   for any $0 < \beta, 1 < \alpha < 1.2$.
\end{proof}

\section{Omitted Proofs of Section~\ref{sec:hard-alpha}}

\subsection{Proof of Lemma~\ref{lem:mom-curve}}
\begin{proof}
Let the sphere be centred at $(a, b, c, d)$ and radius $r$. Consider the function $f(t)=(t-a)^2+\left(t^2-b\right)^2+\left(t^3-c\right)^2+\left(t^4-d\right)^2-r^2$. Note that $t_1$, $t_2$, and $t_3$ are roots of this polynomial. Moreover, $t_2, t_3$ are also roots of $f'(t)$. We will apply Descartes' rule of signs to upper-bound the number of strictly positive roots of $f'(t)$. The rule says that the number of strictly positive roots of a polynomial is upper-bounded by the number of sign changes between non-zero coefficients (assuming the coefficients are arranged in decreasing order of the degree of their corresponding term). To this end, we expand the polynomial $f(t)$ :
$$
\begin{aligned}
f(t) & =t^2-2 a t+a^2+t^4-2 b t^2+b^2+t^6-2 c t^3+c^2+t^8-2 d t^4+d^2-r^2 \\
& =t^8+t^6+(1-2 d) t^4-2 c t^3+(1-2 b) t^2-2 a t+\left(a^2+b^2+c^2+d^2-r^2\right)
\end{aligned}
$$
Thus,
$
f'(t) = 8 t^7+ 6t^5 + 4(1-2 d) t^3 - 6 c t^2 + 2 (1-2 b) t -2 a
$.

Hence, the coefficient sequence is given by $\left(8,6,4(1-2 d),-6 c,2(1-2 b), -2a\right)$. Clearly, there are (at most) 4 changes in sign in this sequence, which implies that the number of strictly positive roots is upper-bounded by 4. However, we already know of 4 roots to this polynomial. Two of them are corresponding to $t_2, t_3$. Using the median value theorem and observing that $f(t_1) = f(t_2) = f(t_3)$, there should exist a $ t'_1 \in (t_1, t_2)$ and $t'_2 \in (t_2, t_3)$ such that $f'(t'_1) = f'(t'_2) = 0$. Therefore $f'$ cannot be zero anywhere else. Note that for $t \rightarrow \infty$, the moment curve must be outside the sphere. Thus, the only way $t_2, t_3, t'_1, t'_2$ could be the only roots of $f'$ is for all $t \in\left(t_1, t_2\right) \cup\left(t_2, t_3\right) \cup (t_3, \infty)$ to lie outside of the sphere.
%
% In particular, since Descartes' rule of signs counts roots of multiplicity separately, the moment curve is not tangent to the sphere for any $t>0$. Now, consider the moment curve in the open interval $\left(t_5, \infty\right)$. It must be the case that the entire curve in this interval lies outside the 3 -sphere. If not, it would have to exit the sphere again at some point, which would result in an additional root (a contradiction). In the following, we imagine going along the curve backwards (i.e., for decreasing values of the parameter $t$ ). For the open interval $\left(t_4, t_5\right)$, since the moment curve is not tangent to the sphere at $t=t_5$, it must go inside the sphere. The next time the curve intersects the 3 -sphere is at $t=t_4$, and hence the curve lies inside the 3 -sphere in the open interval $\left(t_4, t_5\right)$. Similarly, since the curve is not tangent at $t=t_4$, it must exit the 3 -sphere at $t=t_4$ and then intersect the 3 -sphere next at $t=t_3$, implying that the curve lies outside of the 3 -sphere in the open interval $\left(t_3, t_4\right)$. Using the same reasoning, we conclude that the 3 -sphere lies completely inside the 3 -sphere in the open interval $\left(t_2, t_3\right)$, and then completely outside of the 3 -sphere in the open interval $\left(t_1, t_2\right)$, giving the lemma.
\end{proof}

\subsection{Proof of Lemma~\ref{lem:mom-curve-2}}
\begin{proof}
    Similarly to the proof of Lemma~\ref{lem:mom-curve} , let the sphere have centre $(a, b, c)$ and radius $r$. We then analyze the derivative of the function

$$
\begin{aligned}
f(t) & =(t-a)^2+\left(t^2-b\right)^2+\left(t^3-c\right)^2-r^2 \\
& =t^2-2 a t+a^2+t^4-2 b t^2+b^2+t^6-2 c t^3+c^2-r^2 \\
& =t^6+t^4-2 c t^3+(1-2 b) t^2-2 a t+\left(a^2+b^2+c^2-r^2\right)
\end{aligned}
$$
Thus,
$f'(t) = 6 t^5+4 t^3- 6 c t^2+2(1-2 b) t-2 a$.

The coefficients are $\left(6,4,-6 c, 2(1-2 b),-2 a\right)$, which has (at most) 3 changes of sign. Hence Descartes' rule implies that there are at most 3 roots. We already know of 3 roots to this polynomial, two of which correspond to $t_1, t_2$. From the median value theorem, since  $f(t_1) = f(t_2) $, there must exist a $t'_1 \in (t_1, t_2)$ such that $f'(t'_1) = 0$.  Similar to Lemma~\ref{lem:mom-curve}, the only way $t_1, t_2, t'_1$ could be the only roots of $f'$ is for $t \in\left(O, t_1\right) \cup(t_1, t_2) \cup \left(t_2, \infty\right)$ to all lie outside of the sphere.
% But then, since $p_1, \ldots, p_4$ already constitute 4 roots, there are no other roots. Then, the segment $\left(t_4, \infty\right)$ of the moment curve must lie entirely outside the sphere. Furthermore, since the roots are counted with multiplicity, the section $\left(t_3, t_4\right)$ lies inside the sphere, the section $\left(t_2, t_3\right)$ lies outside the sphere, the section $\left(t_1, t_2\right)$ lies inside the sphere, and, finally, the section $\left(O, t_1\right)$ lies outside the sphere.
\end{proof}

\subsection{Proof of Theorem~\ref{thm:k-median-alpha-hardness-panlty}}

\begin{proof}
% Our general approach is reducing the \pvc problem to \kmd with penalties. 
For a fixed parameter $k$, consider a graph $G=(V, E)$ on $n=|V|$ vertices and $m=|E|$ edges, along with an integer $s$. 
Arbitrarily index the vertices $v_1, \ldots, v_n$. 
Construct a Euclidean \kmd instance with penalties in $\bR^4$ denoted by $\cI(G, k)$ as follows. 
% , and cost bound $\nu$ as follows. 
Define $A_3 := \{ (x_1, x_2, x_3, x_4) \in \bR^4: x_4 = 0\}$, i.e., the affine subspace of $\bR^4$ with the fourth coordinate of all points being zero. Every point we define initially lies on $A_3$. 

Consider the 3-dimensional moment curve $\left(t, t^2, t^3\right)$. For each vertex $v_i$, we define $\tilde{v}_i = \left(i ,i^2,i^3, 0 \right) \in A_3$. For each edge $e_{i, j}=\left(v_i, v_j\right)$ in $G$, consider the unique 2-sphere in $A_3$, which we denote by $\bS_{i, j}$, that is perpendicular to the moment curve at points $\tilde{v}_{i}, \tilde{v}_j$. Let $c_{i, j}$ and $r_{i, j}$ denote the centre and radius of the 2-sphere $\bS_{i, j}$, respectively. Let $c_{i, j} = (a, b, c, 0)$. Then the equation system that uniquely solves $a, b, c$ is as follows.
$$
\begin{aligned}
& (i) \quad (i - a)^2+(i^2 - b)^2+(i^3 - c)^2 = (j - a)^2+(j^2 - b)^2+(j^3 - c)^2\\
& (ii)  \quad 6 i^5+4 i^3- 6 c i^2+2(1-2 b) i-2 a = 0\\
& (iii) \quad 6 j^5+4 j^3- 6 c j^2+2(1-2 b) j-2 a = 0
\end{aligned}
$$
Solving the equation system,  we get
\begin{equation} \label{eq:hardness-alpha-k-median-penalty-abc}
    \begin{aligned}
a &= i\cdot j\cdot (i + j)\cdot (3\cdot i^2 + 3\cdot i\cdot j + 3\cdot j^2 + 1) \\
b &= -(3\cdot i^4 + 12\cdot i^3\cdot j + 15\cdot i^2\cdot j^2 + i^2 + 12\cdot i\cdot j^3 + 4\cdot i\cdot j + 3\cdot j^4 + j^2 - 1)/2 \\
c &= (i + j)\cdot (2\cdot i^2 + i\cdot j + 2\cdot j^2 + 1)
\end{aligned}
\end{equation}

Let $q$ be an index pair that gives rise to the maximum $r_{i, j}$, namely $q=\operatorname{argmax}_{i, j} \{r_{i, j}\}$ (i.e., $q$ is of the form " $i, j$ "). Next, define $c^\prime_{i, j} := (a, b, c, \sqrt{r_q^2 - r^2_{i, j}})$ so that $\delta(\tilde v_i, c^\prime_{i, j}) = \delta(\tilde v_j, c^\prime_{i, j}) = r_q$.

As we discussed in the proof of Lemma~\ref{lem:mom-curve-2}, for all $t \in [n]$ we have $\delta( \tilde v_t, c_{i, j})^2 = f(t)$ where $f(t) := (t-a)^2+\left(t^2-b\right)^2+\left(t^3-c\right)^2$. Using Lemma~\ref{lem:mom-curve-2} for any $t \neq i, j$, we have $\delta(\tilde v_t, c_{i, j})^2 = f(t) > r_{i, j}^2$. Examining equation (\ref{eq:hardness-alpha-k-median-penalty-abc}) more closely reveals that $f(t)$ is an integer multiple of $1/4$ for all $t \in [n]$. Thus, we derive $\delta(\tilde v_t, c_{i, j})^2 \geq r_{i, j}^2 + \frac{1}{4}$. Therefore, $\delta(\tilde v_t, c^\prime_{i, j})^2 \geq r_q^2  - r_{i, j}^2 + r_{i, j}^2 + \frac{1}{4} = r_q^2 + \frac{1}{4}$.

Next, we notice that in equation \eqref{eq:hardness-alpha-k-median-penalty-abc}, every coordinate of every $c_{i, j}$ is polynomial with degree at most $5$ with respect to $i$ and $j$, while every coefficient is an integer divided by two. This immediately yields that every $r^2_{i, j} = O(n^{10})$. In particular, $r^2_{q} = O(n^{10})$. 
 % Denoting $\varepsilon = \left(  \sqrt{\frac{r_{q}^2 + \frac{1}{4}}{r_{q}^2} } - 1\right) / 3 $, we would have $\varepsilon = \Omega(\frac{1}{n^{10}})$. Then $\delta\left(c^\prime_{i, j}, \tilde{v}_s\right) \geq (1+ 3 \varepsilon) r_q$. 
  Denoting $\varepsilon = \sqrt{\frac{r_{q}^2 + \frac{1}{4}}{r_{q}^2} } - 1$, we would have $\varepsilon = \Omega(\frac{1}{n^{10}})$. As a result, $\delta\left(c^\prime_{i, j}, \tilde z^* \right) \geq (1+ \varepsilon) r_q$.

 We are finally ready to introduce the data points, candidate centres, and penalty function of $\cI(G, k)$. Define data point to be $X := \{c^\prime_{i, j} \mid i, j \in [n], (v_i, v_j) \in E\}$, candidate centres to be $C := \{\tilde v_i \mid i \in [n]\}$, and penalty function to be $p(c^\prime_{i, j}) := r_q (1 + \varepsilon)$.
\begin{lemma} \label{lem:hardness-alpha-k-median-penalty-cost}
    Let $S = \{\tilde v_{i_1}, \tilde v_{i_2}, ..., \tilde v_{i_k}\} \subseteq X$. Assume that $\{v_{i_1}, v_{i_2}, ..., v_{i_k}\}$ covers $s^*$ edges in $G$. Then $\costp(S) = r^q (m + (m - s^*) \varepsilon)$.
\end{lemma}
\begin{proof}
    Fix any $i, j \in [n]$ where $(v_i, v_j) \in E$. We first calculate the cost of $c^\prime_{i,j}$, which is given by $\min_{\vec t \in S} \left\{ \min\left\{ \delta(c^\prime_{i, j}, \vec t), p(c^\prime_{i, j})\right\} \right\}$. Then, $\costp(S)$ is the sum of the costs of all such $c^\prime_{i, j}$. In case $\tilde v_i$ or $\tilde v_j$ are in $S$ we derive $\min_{\vec t \in S} \left\{ \min\left\{ \delta(c^\prime_{i, j}, \vec t), p(c^\prime_{i, j})\right\} \right\} = \delta(c^\prime_{i, j}, \tilde v_i) = r_{q}$. Otherwise, $\min_{\vec t \in S} \left\{ \min\left\{ \delta(c^\prime_{i, j}, \vec t), p(c^\prime_{i, j})\right\} \right\} = p(c^\prime_{i, j}) = r_{q} ( 1 + \varepsilon)$. Therefore, the total cost of $S$ is $r_q$ multiplied by the number of edges  $\{v_{i_1}, v_{i_2}, ..., v_{i_k}\}$ covers plus $r_q ( 1 + \varepsilon)$ multiplied by the number of edges $\{v_{i_1}, v_{i_2}, ..., v_{i_k}\}$ does not cover. Hence, $\costp(S) = r_q s^* + r_q (m - s^*) (1 + \varepsilon)$. This completes the proof.
\end{proof}

Firstly, Lemma~\ref{lem:hardness-alpha-k-median-penalty-cost} immediately implies that the graph $G$ has a partial vertex cover of size $k$, covering at least $s$ edges, if and only if $\cI(G, k)$ has a solution of cost at most $r^q (m + (m - s) \varepsilon)$. Therefore, the only thing left to prove is that $\cI(G, k)$ is stable. Secondly, from Lemma~\ref{lem:hardness-alpha-k-median-penalty-cost}, we also derive that, assuming that a maximal partial vertex cover of size $k$ of $G$ covers $s^*$ edges, the set of optimum solutions of $\cI(G, k)$ is the set of all $\{\tilde v_{i_1}, \tilde v_{i_2}, ..., \tilde v_{i_k}\}$ such that $\{v_{i_1}, v_{i_2}, ..., v_{i_k}\}$ covers $s^*$ edges in $G$.

% Again from proof of Theorem 4.1 of \cite{cohen2018bane} can be deduced that if the \kmd instance with penalties $(C, X, \delta, P)$ cost is less than certain number, any set $\{v_{i_1}, v_{i_2}, ..., v_{i_k}\} \subseteq V$ is a partial cover for $G$ if and only if $\tilde v_{i_1}, \tilde v_{i_2}, ..., \tilde v_{i_k}, \tilde z^*$ is a solution for $(X, C, \delta, P)$. Therefore, no $f(k) n^{o(k)}$-time algorithm can also solve $(X, C, \delta, P)$ unless \ETH fails. The only thing left to prove is that $(X, C, \delta, P)$ is stable. 

Fix $\varepsilon^{\prime} = \varepsilon / 2m$. We prove that $\cI(G, k)$ is $(1+\varepsilon^{\prime})$-stable 
% w.r.t. $CS^{met-pen}_{mult-mult}$
to complete the proof. Consider any $\delta \leq \delta^\prime \leq (1 + \varepsilon^\prime) \cdot \delta$ and $p \leq p^\prime \leq (1 + \varepsilon^\prime) \cdot p$. Assume the maximal partial vertex cover of size $k$ of $G$ covers $s^*$ edges. We only need to show that any optimum solution of the \kmd instance $(X, C, \delta^\prime, p')$ corresponds to a set of vertices that cover $s^*$ edges. For every $S \subseteq C$ define $\costp^\prime(S) := \sum_{j \in X} \min_{i \in S} \left\{\min \left\{\delta^\prime(i, j), p'(j)\right\}\right\}$. Note that by definition, $\costp(S) \leq \costp^\prime(S) \leq \costp(S) ( 1 + \varepsilon^\prime)$. 

Denote by $O^\prime$ an optimum solution of $(X, C, \delta^\prime, p')$. For the sake of contradiction assume that $O^\prime$ corresponds to a set of vertices that cover $s < s^*$ edges of $G$, and let $O \in \cF_k$ be corresponding to any set of $k$ vertices that cover $s^*$ edges in $G$. By Lemma~\ref{lem:hardness-alpha-k-median-penalty-cost} we have
$$
\begin{aligned}
    \costp^\prime(O) &\leq (1 + \varepsilon^\prime) \costp(O) =  r^q ( m (1 + \varepsilon^\prime)  + (m - s^*) \varepsilon (1 + \varepsilon^\prime)  ) \\
    & <  r^q ( m + \frac{\varepsilon}{2} + (m - s^*) \varepsilon + \frac{\varepsilon^2}{2} )  < r^q ( m  + (m - s^* + 1) \varepsilon )\\
    & \leq r^q ( m  + (m - s) \varepsilon) = \costp(O^\prime) \leq \costp^\prime(O^\prime)
\end{aligned}
$$
This is a contradiction. Therefore, every optimum solution of $(X, C, \delta^\prime, p')$ is also an optimum solution of $(X, C, \delta, P)$, which completes the proof.
\end{proof}

\subsection{Proof of Theorem~\ref{thm:k-median-alpha-hardness}}

\begin{proof}
% Our general approach is reducing the \pvc problem to \xmd{k^\prime}. 
For a fixed parameter $k$, we are given a graph $G=(V, E)$ on $n=|V|$ vertices and $m=|E|$ edges, along with an integer $s$. Arbitrarily index the vertices $v_1, \ldots, v_n$. 
Inspired by Theorem 5.1 of \cite{cohen2018bane}, we construct a Euclidean \xmd{k^\prime} instance in $\bR^6$ with $k^\prime = k + 1$ denoted by $\cI(G, k)$. We emphasize that the proof has been significantly modified to fit our stability setting while also being simplified.
% , and cost bound $\nu$ as follows. 
Define $A_4 := \{ (x_1, x_2, x_3, x_4, x_5, x_6) \in \bR^4: x_5 = 0, x_6 = 0\}$, i.e., the affine subspace of $\bR^4$ with the fifth and sixth coordinates of all points being zero. Every point we define initially lies on $A_4$. 

Consider the 4-dimensional moment curve $\left(t, t^2, t^3, t^4\right)$. For each vertex $v_i$, we define $\tilde{v}_i = \left(i + 1,(i+1)^2,(i+1)^3,(i+1)^4, 0, 0 \right) \in A_4$, and $z^{*} = (1, 1, 1, 1, 0, 0) \in A$. For each edge $e_{i, j}=\left(v_i, v_j\right)$ in $G$, consider the unique 3-sphere in A, which we denote by $\mathbb{S}_{i, j}$, that is perpendicular to the moment curve at point $\tilde{v}_{i}, \tilde{v}_j$ and also has $z^*$ on its surface. Let $c_{i, j}$ and $r_{i, j}$ denote the centre and radius of the 3-sphere $\mathbb{S}_{i, j}$, respectively. Let $c_{i, j} = (a, b, c, d, 0, 0)$ and $\tilde i = (i + 1), \tilde j = (j + 1)$. Then the equation system that uniquely solves $a, b, c, d$ is as follows.
$$
\begin{aligned}
& (i) \quad (\tilde i - a)^2+(\tilde i^2 - b)^2+(\tilde i^3 - c)^2+(\tilde i^4 - d)^2 \\
 & \quad \quad \quad= (\tilde j - a)^2+(\tilde j^2 - b)^2+(\tilde j^3 - c)^2+(\tilde j^4 - d)^2 \\
 & \quad \quad \quad   = (1 - a)^2+(1 - b)^2+(1 - c)^2+(1 - d)^2  \\
& (ii)  \quad a  + 2\tilde i b + 3\tilde i^2 c + 4\tilde i^3 d = \tilde i + 2\tilde i^3 + 3\tilde i^5 + 4\tilde i^7, \\
& (iii) \quad  a  + 2\tilde j b + 3\tilde j^2 c + 4\tilde j^3 d  = \tilde j + 2\tilde j^3 + 3\tilde j^5 + 4\tilde j^7. 
\end{aligned}
$$
One can verify, using any solver system, that the solution to the system of equations is
\begin{equation} \label{eq:hardness-alpha-k-median-abcd}
\begin{aligned}
    a &= -\tilde i\cdot \tilde j\cdot(\tilde i + \tilde j)\cdot(2\cdot \tilde i^3\cdot \tilde j + 4\cdot \tilde i^3 + \tilde i^2\cdot \tilde j^2 + 6\cdot \tilde i^2\cdot \tilde j + 3\cdot \tilde i^2 + 2\cdot \tilde i\cdot \tilde j^3 + 6\cdot \tilde i\cdot \tilde j^2 + 5\cdot \tilde i\cdot \tilde j \\
    & + 4\cdot \tilde i + 4\cdot \tilde j^3 + 3\cdot \tilde j^2 + 4\cdot \tilde j + 2) \\
    b &= (8\cdot \tilde i^5\cdot \tilde j + 4\cdot \tilde i^5 + 17\cdot \tilde i^4\cdot \tilde j^2 + 22\cdot \tilde i^4\cdot \tilde j + 3\cdot \tilde i^4 + 20\cdot \tilde i^3\cdot \tilde j^3 + 34\cdot \tilde i^3\cdot \tilde j^2 + 20\cdot \tilde i^3\cdot \tilde j \\
    & + 4\cdot \tilde i^3 +17\cdot \tilde i^2\cdot \tilde j^4 + 34\cdot \tilde i^2\cdot \tilde j^3 + 29\cdot \tilde i^2\cdot \tilde j^2 + 20\cdot \tilde i^2\cdot \tilde j + 2\cdot \tilde i^2 + 8\cdot \tilde i\cdot \tilde j^5 + 22\cdot \tilde i\cdot \tilde j^4  \\
    & + 20\cdot \tilde i\cdot \tilde j^3 + 20\cdot \tilde i\cdot \tilde j^2+ 8\cdot \tilde i\cdot \tilde j + 4\cdot \tilde j^5 + 3\cdot \tilde j^4 + 4\cdot \tilde j^3 + 2\cdot \tilde j^2 + 1)/2 \\
    c &= -(\tilde i + \tilde j)\cdot (2\cdot \tilde i^4 + 6\cdot \tilde i^3\cdot \tilde j + 4\cdot \tilde i^3 + 5\cdot \tilde i^2\cdot \tilde j^2 + 6\cdot \tilde i^2\cdot \tilde j + 4\cdot \tilde i^2 + 6\cdot \tilde i\cdot \tilde j^3 + 6\cdot \tilde i\cdot \tilde j^2 \\
    &+ 7\cdot \tilde i\cdot \tilde j + 4\cdot \tilde i  + 2\cdot \tilde j^4 + 4\cdot \tilde j^3 + 4\cdot \tilde j^2 + 4\cdot \tilde j + 2)\\
    d &= (5\cdot \tilde i^4 + 8\cdot \tilde i^3\cdot \tilde j + 4\cdot \tilde i^3 + 9\cdot \tilde i^2\cdot \tilde j^2 + 6\cdot \tilde i^2\cdot \tilde j + 6\cdot \tilde i^2 + 8\cdot \tilde i\cdot \tilde j^3 + 6\cdot \tilde i\cdot \tilde j^2 + 8\cdot \tilde i\cdot \tilde j \\ &
    + 4\cdot \tilde i + 5\cdot \tilde j^4 + 4\cdot \tilde j^3 + 6\cdot \tilde j^2 + 4\cdot \tilde j + 3)/2 \\
\end{aligned}
\end{equation}

Similar to Theorem~\ref{thm:k-median-alpha-hardness-panlty}, let $q:=\operatorname{argmax}_{i, j} \{r_{i, j}\}$ and define $c^\prime_{i, j} := (a, b, c, d, \sqrt{r_q^2 - r^2_{i, j}}, 0)$ so that $\delta(z^*, c^\prime_{i, j}) = \delta(\tilde v_i, c^\prime_{i, j}) = \delta(\tilde v_i, c^\prime_{i, j}) = r_q$. Again for any $t \neq i, j \in [n]$, we have $\delta(\tilde v_t, c^\prime_{i, j})^2 \geq r_q^2 + \frac{1}{4}$. 
% We also define 
% $$\tilde{z}^{*} = \left(1, 1, 1, 1, 0, \sqrt{\sqrt{r^2_q + \frac{1}{4}}\cdot r_q - r^2_q}\right)$$ such that 
% \begin{equation}
%     \frac{\delta(\tilde{z}^*, c^\prime_{i, j})}{\sqrt{r^2_q + \frac{1}{4}}} = \frac{r_q}{\delta(\tilde{z}^*, c^\prime_{i, j})}.
% \end{equation}
We also define 
$\tilde{z}^{*} = \left(1, 1, 1, 1, 0, \frac{1}{2} \right)$, and subsequently $\delta^2(\tilde z^*, c_{i, j}) = r^2_q + \frac{1}{4}$.

Moreover, in \eqref{eq:hardness-alpha-k-median-penalty-abc} every coordinate of every $c_{i, j}$ is a polynomial with degree at most $7$ with respect to $\tilde i$ and $\tilde j$, while every coefficient is an integer divided by two. This immediately yields that $r^2_{i, j} = O(n^{14})$. In particular, $r^2_{q} = O(n^{14})$. 
Setting 
% $$\varepsilon =  \sqrt{\frac{\sqrt{r^2_q + \frac{1}{4}}}{r_{q}} } - 1, $$
% we would have $\delta(z^*, c^\prime_{i, j}) = (1 + \varepsilon) r_q$ and  $\delta(\tilde{v}_s, c^\prime_{i, j}) \geq  (1 + \varepsilon) \delta(z^*, c^\prime_{i, j})$ for all $s \neq i, j \in [n]$. Also $\varepsilon = \Omega(\frac{1}{n^{14}})$.  
$\varepsilon =  \sqrt{\frac{r_q^2 + \frac{1}{4}}{r_{q}^2} } - 1, $
we would have $\delta(z^*, c^\prime_{i, j}) = (1 + \varepsilon) r_q$ and $\varepsilon = \Omega(\frac{1}{n^{14}})$.

 We are finally ready to introduce the data points and candidate centres of $\cI(G, k)$. Define the set of candidate centres to be $C := \{\tilde v_1, \tilde v_2, ..., \tilde v_n, \tilde z^*\}$, and the set of data points $X$ to be $\{c^\prime_{i, j} \mid i, j \in [n], (v_i, v_j) \in E\}$ in addition to $\lceil m r_q \rceil$ points in $\tilde z^*$.

\begin{lemma} \label{lem:hardness-alpha-k-median-cost}
    Let $S = \{\tilde v_{i_1}, \tilde v_{i_2}, ..., \tilde v_{i_k}\} \cup \{ \tilde{z}^* \} \subseteq C$. Assume that $\{v_{i_1}, v_{i_2}, ..., v_{i_k}\}$ covers $s^*$ edges in $G$. Then $\cost(S) = r^q (m + (m - s^*) \varepsilon)$.
\end{lemma}
\begin{proof}
    Note that all data points in $\tilde z^*$ will be assigned to $\tilde z^*$ and their cost would be zero. Consider any $(v_i, v_j) \in E$. Then, if $\tilde v_i$ or $\tilde v_j$ is in $S$, we have $\min_{\vec t \in S} \delta(c^\prime_{i, j}, \vec t) = r_{q}$. Otherwise, $\min_{\vec t \in S} \delta(c^\prime_{i, j}, \vec t) = \delta(c^\prime_{i, j}, \tilde z^*) = r_{q} ( 1 + \varepsilon)$. This completes the proof.
\end{proof}

\begin{lemma} \label{lem:hardness-alpha-k-median-solutions}
    Assuming a maximal partial vertex cover of size $k$ of $G$ covers $s^*$ edges, the set of optimum solutions of $\cI(G, k)$ is the set of all $\{\tilde v_{i_1}, \tilde v_{i_2}, ..., \tilde v_{i_k}, \tilde z^*\}$ such that $\{v_{i_1}, v_{i_2}, ..., v_{i_k}\}$ covers $s^*$ edges in $G$.
\end{lemma}
\begin{proof}
    % First we need to prove that every optimum solution of $\cI(G, k)$ contains $\tilde z^*$. For the sake of contradiction assume that $O \subseteq \cF_{k^\prime}$ is an optimum solution that does not have $\tilde z^*$.
    First consider any $O \in \cF_{k^\prime}$ that does not contain $\tilde z^*$. We will prove that such $O$ cannot be an optimum solution. The closest centre to $\tilde z^*$ is $\tilde v_1$ and $\delta(\tilde z^*, \tilde v_1) \geq 2$. If all data points in $z^*$ were assigned to $\tilde v_1$, the cost would be $2 \lceil m r_q \rceil$. Therefore, due to Lemma~\ref{lem:hardness-alpha-k-median-cost}, $\cost(O)$ would be bigger than the cost of any $S \in \cF_{k^\prime}$ that contains $\tilde z^*$, which is a contradiction. The rest of the proof immediately follows from Lemma~\ref{lem:hardness-alpha-k-median-cost}.
\end{proof}

Lemma~\ref{lem:hardness-alpha-k-median-cost} and Lemma~\ref{lem:hardness-alpha-k-median-solutions} imply that the graph $G$ has a partial vertex cover of size $k$ covering at least $s$ edges if and only if $\cI(G, k)$ has a solution of cost at most
$r^q (m + (m - s) \varepsilon)$. Therefore, the only thing left to prove is that $\cI(G, k)$ is stable.

% Again from proof of Theorem 4.1 of \cite{cohen2018bane} can be deduced that if the \kmd instance with penalties $(C, X, \delta, P)$ cost is less than certain number, any set $\{v_{i_1}, v_{i_2}, ..., v_{i_k}\} \subseteq V$ is a partial cover for $G$ if and only if $\tilde v_{i_1}, \tilde v_{i_2}, ..., \tilde v_{i_k}, \tilde z^*$ is a solution for $(X, C, \delta, P)$. Therefore, no $f(k) n^{o(k)}$-time algorithm can also solve $(X, C, \delta, P)$ unless \ETH fails. The only thing left to prove is that $(X, C, \delta, P)$ is stable. 

Fix $\varepsilon^{\prime} = \varepsilon / 2m$. We prove that $\cI(G, k)$ is $(1 + \varepsilon^{\prime})$-stable, which completes the proof. Consider any 
% $\delta^\prime \in \CS^{met}_{mult}(\delta, P, 1 + \varepsilon^\prime)$
$\delta \leq \delta^\prime \leq (1 + \varepsilon^\prime) \cdot \delta$ and $p \leq p' \leq (1 + \varepsilon^\prime) \cdot p$. Assume that a maximal partial vertex cover of size $k$ of $G$ covers $s^*$ edges. It suffices to show that any optimum solution of $(X, C, \delta^\prime)$ corresponds to a set of vertices that covers $s^*$ edges. For every $S \in \cF_{k^\prime}$, define $\cost^\prime(S) = \sum_{j \in X} \min_{i \in S} \delta^\prime(i, j)$. Note that by definition $\cost(S) \leq \cost^\prime(S) \leq \cost(S) ( 1 + \varepsilon^\prime)$. 

Let $O^\prime$ be an optimum solution of the \xmd{k^\prime} instance $(X, C, \delta^\prime)$.  For the sake of contradiction assume that $O^\prime$ corresponds to a set of vertices that cover $s < s^*$ edges of $G$, and let $O$ correspond to any set of vertices that cover $s^*$ edges in $G$. Similarly to the proof of Theorem~\ref{thm:k-median-alpha-hardness-panlty}, and using Lemma~\ref{lem:hardness-alpha-k-median-penalty-cost}, we derive
$$
\begin{aligned}
    \cost^\prime(O) &\leq (1 + \varepsilon^\prime) \cost(O)  < r^q \left( m  + (m - s^* + 1) \varepsilon \right) \leq r^q ( m  + (m - s) \varepsilon) \\
    & = \cost(O^\prime)  \leq \cost^\prime(O^\prime)
\end{aligned}
$$
This is a contradiction. Therefore, every optimum solution of $(X, C, \delta^\prime)$ is also an optimum solution of $(X, C, \delta)$, which completes the proof.
\end{proof}

\end{document}